\documentclass[11pt]{article}
\usepackage{graphicx}
\usepackage{geometry}
\usepackage{adjustbox}
\usepackage{dsfont}
\providecommand{\keywords}[1]{\textbf{\textit{Keywords---}} #1}
\title{
Entropy of Exchangeable Random Graphs}
\author{
    Anda Skeja$^{1,2}$\footnote{This work was supported by the European Research Council under Grant CoG 2015-682172NETS, within the Seventh European Union Framework Program and the Swiss Confederation by the Excellence Scholarship for Foreign Scholars.} \quad Sofia C. Olhede$^{1*}$ \\  
    {\small $^1$Institute of Mathematics, EPFL} \\
    {\small $^2$Department of Mathematics, Uppsala University}}\usepackage[utf8]{inputenc} 
\usepackage{algorithm}
\usepackage{bbm}

\usepackage[T1]{fontenc}
\usepackage{url}              
\usepackage{cite}            
\usepackage[cmex10]{amsmath} 
\usepackage{mleftright}       
\mleftright                   
\numberwithin{equation}{section}

\usepackage{graphicx}        
\usepackage{booktabs}         
\usepackage{amsmath,amsfonts}
\usepackage{algorithmic}
\usepackage{array}
\usepackage[caption=false,font=normalsize,labelfont=sf,textfont=sf]{subfig}
\usepackage{textcomp}
\usepackage{stfloats}
\usepackage{url}
\usepackage{verbatim}
\usepackage{graphicx}
\usepackage{cite}
\usepackage{amsthm}
\usepackage{booktabs}
\usepackage[font={small,it}]{caption}
\usepackage{accents}

\usepackage[pdfstartview=XYZ,
bookmarks=true,
colorlinks=true,
linkcolor=blue,
urlcolor=blue,
citecolor=blue,
pdftex,
bookmarks=true,
linktocpage=true, 
]{hyperref}

\usepackage{orcidlink}

\usepackage{graphicx}

\newcommand{\rtwo}{\rho^{(2)}}

\newcommand{\kap}[1]{\rtwo(1)-\lambda^2}

\definecolor{thecolor}{rgb}{0.1,0.5,1}

\newtheorem{theorem}{Theorem}
\newtheorem{proposition}{Proposition}
\newtheorem{cor}{Corollary}
\newtheorem{lemma}{Lemma}
\newtheorem{definition}{Definition}

\definecolor{darkgreen}{rgb}{0.0, 0.2, 0.13}  

\date{}
\begin{document}
\maketitle
\begin{abstract}
Quantifying the complexity of large graphs requires measures that extend beyond predefined structural features and scale efficiently with graph size. This work adopts a generative perspective, modeling large networks as exchangeable graphs to quantify the information content of their generating mechanisms via graphon entropy. As a graph property, graphon entropy is invariant under isomorphisms, making it an effective measure of complexity; however, it is not directly computable. To address this, we introduce a suite of graphon entropy estimators, including a nonparametric estimator for broad applicability and specialized versions for structured graphons arising from well-studied random graph models such as Erd\H{o}s–R\'enyi, Chung–Lu, and stochastic block models. We establish their large-sample properties, deriving convergence rates and Central Limit Theorems. Simulations illustrate how the nonparametric graphon entropy estimator captures structural variations in graphs, while real-world applications demonstrate its role in characterizing evolving network dynamics.

\end{abstract}

\section{Introduction}
Large graphs possess intricate structures shaped by the interplay of order and randomness. Quantifying their complexity is essential for understanding the underlying patterns. While complexity lacks a universal definition, a common approach is to measure a graph’s information content. Entropy, a key concept in information theory \cite{shannon1948mathematical}, provides a rigorous framework for this purpose and has been widely extended to graphs (see Section~\ref{related-work} for a review). However, these measures often rely on predefined features or scale poorly with graph size. The challenge is particularly evident in graphs that are not deterministic, where connectivity patterns exhibit variability that cannot be easily reduced to a fixed set of descriptors.

To overcome this limitation, we shift our focus from structural summaries to the generative process of graphs. In doing so, we adopt a well-established framework for modeling large graphs: \emph{exchangeability}, where edges are assumed to be generated according to a symmetric measurable function known as a \emph{graphon}~
\cite{lovasz2012large}. Graphons offer a probabilistic representation of large networks, naturally encoding both local and global structure.  The entropy of a graphon, introduced in theoretical works by~\cite{janson2013graphons,hatami2018graph}, offers a principled approach to quantifying network complexity from a generative perspective. However, its potential as a complexity measure remains largely unexplored, and methods for estimating graphon entropy from a single observed network are lacking. 

\textbf{Our contributions.} We develop a suite of graphon entropy estimators adaptable to diverse network structures. Our primary approach is a nonparametric estimator, providing a flexible method for graphon entropy estimation in real-world networks while addressing the important aspect of bandwidth selection. To optimize performance for structured graphon classes, we develop specialized entropy estimators and establish their convergence rates and Central Limit Theorems (CLTs). We begin with entropy estimation under a constant graphon, corresponding to the limiting behavior of Erdős–Rényi graphs~\cite{erdds1959random}, where only the marginal edge probability is considered. In this case, we derive a CLT for the entropy estimator using standard asymptotic arguments. We then extend this framework to separable graphons, which correspond to the limits of Chung–Lu~\cite{chung2002average} graph sequences and incorporate degree-dependent edge probabilities. In this setting, we derive a more general CLT, addressing structural dependencies using U-statistics. Next, we analyze block-constant graphons, corresponding to stochastic block models~\cite{holland1983stochastic}, and establish convergence rates for entropy estimation in this setting. To assess the practical applicability of our approach, we conduct an extensive simulation study, evaluating the performance of our estimators, comparing them across different graphon structures, and analyzing the relationship between graphon entropy and graph topology. Additionally, we compare the nonparametric graphon entropy estimator with traditional network entropy measures, highlighting its distinct advantages in capturing structural complexity. Finally, we apply the nonparametric graphon entropy estimator to real-world temporal graphs, demonstrating its effectiveness in capturing the evolution of networks over time.

Our approach fundamentally differs from classical entropy estimation methods and other graph entropies reviewed in Section~\ref{related-work}. Traditional entropy estimators typically assume i.i.d. samples~\cite{gyorfi1991nonparametric,beirlant1997nonparametric,verdu2019empirical}, whereas we estimate entropy from conditionally independent edges within a single observed network. As previously discussed, our framework also differs from existing graph entropy measures in that we focus on the graph-generating mechanism.
 
The remainder of this paper is structured as follows. Section~\ref{sec-prelim} introduces the necessary notation and background. Section~\ref{sec-estimation-arbitrary} develops the estimation of graphon entropy for arbitrary exchangeable graphs and establishes its convergence properties. Section~\ref{sec:est-random-graphs} examines the estimation of graphon entropy for common random graph models and analyzes their asymptotic behavior. Section~\ref{Sec:SimSt} presents a comprehensive simulation study evaluating estimator performance, its relationship with network structure, and comparisons with network entropy measures. Section~\ref{Sec:RealData} applies our nonparametric estimator to real-world networks, demonstrating its effectiveness in tracking structural evolution. Finally, Section~\ref{Sec:Disc} concludes with a discussion of our findings.

\subsection{Related Work}\label{related-work}

\subsubsection{Entropy Measures for Networks}
Entropy as a topological measure of the information content of graphs was first introduced by~\cite{rashevsky1955life} in 1955 and was picked up in subsequent works~\cite{trucco1956note,mowshowitz1968entropy1,mowshowitz1968entropy2,mowshowitz1968entropy3,mowshowitz1968entropy4,mowshowitz2012entropy}. A slightly different approach was introduced by~\cite{korner1973coding}, where graph entropy indicates the number of bits that must be transmitted to clear up the ambiguity of a graph vertex. Yet another graph entropy, the so-called chromatic entropy, was proposed by~\cite{alon1996source}. An information-theoretic entropy definition for graph complexity was first introduced by~\cite{mowshowitz1968entropy1}, where it was posited that an effective complexity index for a mathematical object should be based on a partial set of invariants shared by isomorphic graphs. In particular, this set was identified as the automorphism group of the graph. These entropy measures are suited for small graphs, i.e. having fewer than thirty nodes. Over time, the definition of entropy has evolved to encompass additional graph features such as vertices, edges, and degree distributions, which has received criticism for being feature-dependent~\cite{zenil2014correlation} and not capturing the whole graph structure. Similarly, algorithmic complexity (Kolmogorov complexity)~\cite{kolmogorov1968logical} is theoretically appealing because it captures the shortest possible description of a graph, but its incomputability makes it impractical. In practice, approximations via compression-based methods often fail to scale efficiently with increasing graph size. The entropy defined by the random walker approach~\cite{small2013complex} captures the uncertainty in navigation across the network, while the Kolmogorov-Sinai entropy~\cite{demetrius2005robustness} quantifies the overall complexity by measuring the rate of information loss in its dynamics. For related work defining feature-dependent entropy measures for graphs, see~\cite{bianconi2009entropy,aziz2021graph,dehmer2008novel,ye2014approximate}.
A review of the graph entropy measures mentioned above is covered in~\cite{dehmer2011history}. Furthermore, since entropy represents a mathematical limit on how well data from the source can be losslessly compressed, it has been studied in the context of compression for various random graph models. The Shannon entropy of the isomorphism class of a random graph was called the ``structural entropy'' and was asymptotically characterized for the Erd\H{o}s-R\'enyi model in~\cite{choi2012compression}. The same measure has been asymptotically characterized for preferential attachment graphs; see~\cite{luczak2019compression,luczak2019asymmetry}. There is also a large literature on universal compression of graphs from various distribution classes~\cite{zhang2013universal,ganardi2019universal,delgosha2022universal,bhatt2023universal,bianconi2009entropy,paton2022entropy}. The entropy of the ensembles of stochastic block models has been studied from a statistical mechanics approach by~\cite{peixoto2012entropy}. The Shannon entropy for network ensembles and its connection with the Gibbs and von Neumann entropies were discussed by~\cite{anand2009entropy}. The von Neumann entropy for the graphs was introduced by~\cite{braunstein2006laplacian}, and~\cite{passerini2008neumann} linked it to Shannon entropy of degree sequences, suggesting higher entropy for heterogeneous networks, though its relation to community structure remains unclear. While structured graphs show lower entropy, networks with weakly separated communities may deviate~\cite{de2016spectral}, leaving its full interpretation unsettled. Authors of~\cite{bordenave2015large} introduced an entropy functional for studying large deviations in Erd\H{o}s-R\'enyi graphs. 

\subsubsection{Nonparametric Entropy Estimation} Entropy is one of the most extensively studied functionals of probability density functions, with significant research focused on its nonparametric estimation. For a comprehensive discussion on entropy estimation and nonparametric techniques, see~\cite{beirlant1997nonparametric, paninski2003estimation}. These techniques include methods based on histograms as studied by~\cite{hall1993estimation}, kernel density estimates by~\cite{paninski2008undersmoothed}, sample spacings in the univariate case by~\cite{haje2009entropy}, and more recently weighted estimators~\cite{moon2016nonparametric,sricharan2013ensemble,Berrett2019EntropyEstimation}. We focus on entropy estimation studies most relevant to our research rather than providing an exhaustive list for i.i.d. observations. Recent works on entropy estimation of discrete probability distributions include maximum likelihood and minimax estimation of functionals, see~\cite{jiao2017maximum,jiao2015minimax}, where the authors point out that the maximum likelihood estimator of the entropy is sub-optimal in sample complexity, i.e. it is sub-optimal in the large alphabet--small sample setting and is efficient in large sample settings. For entropy estimation based on kernel density estimation, numerous studies have been conducted; for example, ~\cite{liu2012exponential,paninski2008undersmoothed} among others argue for under-smoothing during entropy estimation. In other words, it is argued that using a smaller bandwidth than that required for consistent kernel density estimation yields consistent entropy estimates. 

\section{Notation and Background}\label{sec-prelim}
All graphs in this paper are undirected and simple, i.e. without self-loops and multiple edges. For a graph $G$, $V(G)$ denotes the set of vertices or nodes, and $E(G)$ denotes the set of edges. We say two nodes $i$ and $j$ are adjacent if they are connected by an edge. Given a graph $G$ with $|V(G)|=n$, we denote its adjacency matrix by $A\in \{0,1\}^{n \times n}$, where $A_{ij}$ equals one if nodes $i,j\in [n]$ are adjacent, and zero otherwise. We denote the degree of vertex $i$ by $d_i$. We let $\log$ denote the natural logarithm. See Appendix D.0 for preliminaries.
 
 \subsection{Graphon models}
 This section introduces the graphon models used throughout the paper and states the representation theorem by Aldous and Hoover~\cite{aldous1981representations,hoover1979relations}. These models serve as limiting objects for various classes of random graphs.
\begin{definition}[Constant graphon]\label{def-er}
   A constant graphon is a function $W^c:[0,1]^2\to[0,1]$ defined by $$W^c(x,y)=p, \quad \forall (x,y) \in [0,1]^2, \quad p \in [0,1].$$ It represents the limiting object of a sequence of Erdős–Rényi random graphs.
\end{definition}
\begin{definition}[Separable graphon]\label{def-cm} A separable graphon is a function \( W^s: [0,1]^2 \to [0,1] \) defined by  
\[
W^s(x,y) = g(x) g(y), \quad \forall (x,y) \in [0,1]^2,
\]
for some function \( g: [0,1] \to [0,1] \). It corresponds to the limit of a sequence of Chung–Lu graphs.

\end{definition}
\begin{definition}[Block-constant graphon]\label{def-sbm} A block-constant graphon is a function \( W^b: [0,1]^2 \to [0,1] \) defined by  
\[
W^b(x,y) = \sum_{a=1}^k\sum_{b=1}^k\theta_{ab}\mathds{1}((x,y)\in \omega_{ab}),
\]
where \( k \leq n \) is the number of blocks, and  
$
\omega_{ab}=\{(x,y):\, \Tilde{H}(a-1)\leq x< \tilde{H}(a), \  \tilde{H}(b-1)\leq y< \tilde{H}(b)\}.
$
The function \( \tilde{H}(\cdot) \) is the aggregate block width function, defined as  
\[
\tilde{H}(u)=\sum_{a=1}^{\lfloor u \rfloor} \frac{1}{n}h_a,\quad u\in [0,k], \; \tilde{H}(u)\in \Big\{0,\frac{h_1}{n},\frac{h_1+h_2}{n},\ldots,1\Big\},
\]
where \( h_a \) is the number of nodes in group \( a \in [k] \).  
It corresponds to the limit of a sequence of stochastic block models.

\end{definition}
\begin{definition}\label{assump:smooth}($\alpha$-H\"older graphon) An \(\alpha\)-Hölder graphon is a function \( W: [0,1]^2 \to [0,1] \) such that  
 
\[
W\in\mathrm{H\ddot{o}lder}^{\alpha}(M) \iff \sup_{(x,y)\neq(x',y')\in(0,1)^2} \frac{|W(x,y)-W(x',y')|}{|(x,y)-(x',y')|^{\alpha}}\leq M<\infty.
\]
\end{definition}
    Throughout this paper, we will use the terms $\alpha-$H\"older continuous and smooth interchangeably.
\begin{theorem}[Aldous--Hoover~\cite{aldous1981representations,hoover1979relations}]\label{thm-aldous}
Let $A$ be a jointly exchangeable random array. Then there exists an i.i.d. sequence $\xi=(\xi_1,...,\xi_n)$ following $U(0,1)$, a random variable $\gamma \sim U(0,1)$ independent of $\xi$, and a function $W:[0,1]^3 \to [0,1]$ such that \begin{equation}\label{graphon}
    Pr(A_{ij}=1 | \xi, \gamma) = W(\xi_i,\xi_j,\gamma),
\end{equation} and $A_{ij}$ are conditionally independent across $i,j$ given $\xi$ and $\gamma$. 
\end{theorem}\label{thm-aldous-hoover}
Since we observe only one realization of \( A \), the dependence on \( \gamma \) cannot be estimated and is suppressed, assuming the graph is disassociated. Thus, we use \( W(\xi_i,\xi_j) \) for \( 1 \leq i,j \leq n = |V| \) throughout.
\subsection{Graphon entropy}
Graphon entropy is a graph property and is determined solely by the underlying graph limit~\cite{janson2013graphons,hatami2018graph}. It is the same for any two equivalent $[0,1]-$valued graphons. This makes it well-defined for every equivalence class of such graphons~\cite{janson2016graphons}. 
\begin{definition}[Entropy of a graphon~\cite{janson2013graphons,hatami2018graph}]\label{def-entropy-hatami}
We define the binary entropy function $h:[0,1]\rightarrow \mathbb{R}_{+}$ to be
\begin{equation}
    \label{eqn:binary}
    h(x)=-x\log(x)-(1-x)\log(1-x),
\end{equation}\label{binaryentropy}
for $x\in [0,1]$ where $h(0)$ and $h(1)$ are defined to be zero to ensure continuity. The entropy of a graphon \( W(x,y) \) is defined as  
\begin{equation}
\label{eqn:Hatemi-ent}
\mathrm{H}(W) = \iint_{[0,1]^2} h(W(x,y))\,\mathrm{d}x\,\mathrm{d}y.     
\end{equation}
\end{definition}

Graphon entropy quantifies the balance between structure and randomness in graphs, offering insight into different graph models. For example, a \textit{random-free} graphon, which takes values only in $\{0,1\}$ almost everywhere, represents an extreme case of predictability with minimal entropy. At the other extreme, highly random graphs, such as the Erdős–Rényi model with edge probability $1/2$, maximize entropy at 1 (base-2 logarithm). This occurs when the graphon function is constant, $W(x,y) = 1/2$, indicating complete edge independence and forming the limit of \textit{quasi-random} graph sequences~\cite{hatami2018graph,lovasz2006limits}. Between these extremes, structured models like stochastic block models and Chung–Lu graphs introduce patterns that systematically reduce entropy. This continuum—from highly ordered to maximally disordered structures—underscores the relevance of graphon entropy in characterizing complexity in network analysis.  See Appendix D.0 for a discussion on the relationship between the entropy of an exchangeable graph and graphon entropy. 

Furthermore, since entropy is the fundamental building block in information theory, graphon entropy provides a foundation for developing multivariate complexity measures and their estimators across multiple networks, as explored in subsequent work (e.g.,~\cite{skeja2024quantifying}). 

\subsection*{Sparse Regime}

We will also consider scaled graphons, aiming to estimate the graphon entropy of an exchangeable sparse graph, as described in Equation~\eqref{sparse-graphon} (Appendix D.0). Given that $0 < \rho_n |W|_{\infty} < 1$ for $0 < \rho_n \leq 1$, the entropy of the scaled graphon is defined as follows:
\begin{equation}\label{sparse-ent}
\mathrm{H}(f)=\iint_{[0,1]^2} h( f(x,y))\ \mathrm{d}x \mathrm{d}y,
\end{equation}
where $f(x,y)=\rho_n W(x,y)$. Throughout this paper, we use $f$ to denote the scaled graphon across various sparsity regimes.

Having discussed graphon entropy, including its scaled variant, the next step is to estimate it from observed networks. In Section 3, we introduce a nonparametric graphon entropy estimator to characterize graph complexity from a generative perspective.
\section{Graphon Entropy Estimator of an Exchangeable Graph With a Smooth Graphon Representative }\label{sec-estimation-arbitrary}

Graphon entropy, computed by integrating over the unit square, remains invariant under interval permutations and measure-preserving transformations, making it an effective measure of structural complexity. To estimate it from a single network observation, we assume an underlying smooth graphon representation necessary for tractable estimation. We approximate this graphon using a stochastic block model with a progressively increasing number of blocks~\cite{gao2015rate,klopp2017oracle,Olhede_2014} (see Equation~\eqref{est-graphon} in Appendix D.0). We then compute the entropy of the estimated graphon, obtaining an estimator that quantifies the network’s inherent randomness and structure. The resulting entropy estimator is given by:

\begin{equation}\label{entropy-estimator} \widehat{\mathrm{H}}(f)=-\iint_{(0,1)^2}\hat{f}(x,y;h) \log(\hat{f}(x,y;h)) + (1-\hat{f}(x,y;h))\log(1-\hat{f}(x,y;h)) \ \mathrm{d}x \mathrm{d}y, \end{equation} where $\hat{f}(x,y;h)$ is defined in Equation~\eqref{est-graphon} (Appendix D.0). Similarly any consistent graphon estimator could be used for this step. Our nonparametric graphon entropy estimator satisfies the convergence rates detailed in Theorem~\ref{thm-conv-rate-entropy} for $k=\lceil{\rho_n^{1/2} n^{1/(1+\alpha)}}\rceil$.
In the following subsection, we address the important aspect of bandwidth selection for the graphon entropy estimator.
\begin{theorem}
\label{thm-conv-rate-entropy} Let the estimator of $\mathrm{H}(f)$ from a single realization of $G(n,f)$ be obtained as described in~\eqref{entropy-estimator}, where $0<\rho_n\leq 1$. Then, for any constant $C>0$ (with $C_1>0$ being a function of $C$), there is a constant $C'>0$ depending only on $M$ and $C$ such that the following inequality 
   \begin{align*}
    \small
        |\widehat{\mathrm{H}}(f)-\mathrm{H}(f)|&\leq  C_1\left(\rho_n^{-2\alpha/(\alpha+1)}n^{-2\alpha/(\alpha+1)}+\frac{\rho_n \log n}{n}+\frac{\rho_n^2}{n^{\alpha}}\right)^{1/2}\\
&+\mathcal{O}\left(\rho_n^{-2\alpha/(\alpha+1)}n^{-2\alpha/(\alpha+1)}+\frac{\rho_n \log n}{n}+\frac{\rho_n^2}{n^{\alpha}}\right),
    \end{align*}
    holds with probability at least $1-\exp(-C'n)$, uniformly over $f\in \mathcal{F}_{\alpha}(M)$.
\end{theorem}
\begin{proof}
    See Appendix A.1.
\end{proof}
Nonparametric entropy estimation for graph data differs fundamentally from its application to i.i.d. samples, rendering traditional methods ineffective. To serve as a meaningful complexity measure, the graphon entropy estimator relies on an initial graphon estimation step, which cannot be excluded. While resembling density estimation, graphon estimation distinctively incorporates a clustering step. Unlike kernel density-based entropy estimators, which require bandwidth reduction for nonparametric convergence~\cite{paninski2008undersmoothed}, Proposition~\ref{prop-oracle} shows that the optimal bandwidth for graphon entropy cannot be too small due to clustering challenges~\cite{choi2012stochastic}.

\subsection{Oracle Graphon Entropy Estimator} 
Equation~\eqref{entropy-estimator} relies on the bandwidth $h$, which is a necessary parameter for accurate estimation. While automatic bandwidth selection is discussed in the context of nonparametric graphon estimation~\cite{Olhede_2014}, we extend this oracle-based approach to determine the optimal bandwidth for the graphon entropy estimator and justify its applicability to the empirical estimator. The oracle estimator assumes knowledge of the latent variables $\xi_i$ ($1 \leq i \leq n$), which are otherwise unknown, to bound the performance of data-driven estimators. It mirrors Equation~\eqref{entropy-estimator} but uses a unique labeling $\Tilde{z}$ derived from the latent vector $\xi$. Following~\cite{Olhede_2014}, we have $\Tilde{z} = \min {\lceil (i)^{-1}/h \rceil, k}$, where $(i)^{-1}$ ranks the $i^{th}$ element of $\xi$ in ascending order, grouping indices into sets of size $h$, with the last set adjusted for remainders. 

\begin{definition}[Oracle entropy estimator] Based on the explanation above, for $\Tilde{z}$ the oracle membership vector, the oracle entropy estimator takes the following form
\begin{equation}
    \label{oracle-entropy-estimator}
     \widehat{\mathrm{H}}^*(f)=-\iint_{(0,1)^2}\hat{f}^*(x,y) \log(\hat{f}^*(x,y)) + (1-\hat{f}^*(x,y))\log(1-\hat{f}^*(x,y)) \  \mathrm{d}x \mathrm{d}y,
\end{equation}
\end{definition}
where $\hat{f}^*(x,y)=\rho_n^{+}\Bar{A}^*_{\lceil xk\rceil,\lceil yk\rceil}(\tilde{z})$ follows the same form as Equation~\eqref{est-graphon} of Appendix D.0, with $\hat{z}$ replaced by the oracle clustering $\tilde{z}$.
Specifically,\begin{equation}
\Bar{A}_{ab}^*(\tilde{z})=\frac{\sum_{i<j}A_{ij}\mathds{1}(\Tilde{z}_i=a)\mathds{1}(\Tilde{z}_j=b)}{\sum_{i<j}\mathds{1}(\Tilde{z}_i=a)\mathds{1}(\Tilde{z}_j=b)},
\label{oracle-A}
\end{equation}
for $1\leq a,b \leq k$.

We now provide the bias and variance from the oracle graphon entropy estimator to discuss automatic bandwidth selection for graphon entropy estimation. 
    
\begin{proposition}\label{prop-oracle}Let $G(n,f)$  be generated according to Definition~\ref{assump:smooth}, and consider the graphon entropy estimated using Equation~\eqref{oracle-entropy-estimator} with the oracle estimator of the graphon. The bias and variance of the oracle graphon entropy estimator are given by:
\begin{align*}
    |\mathbb{E}[{\mathcal{\widehat{\mathrm{H}}}^*(f)}]-\mathrm{H}(f)|&\leq \left(\frac{1}{\rho_n\min \bar{f}_{a,b}}+\frac{1}{\min (1-\rho_n\bar{f}_{a,b})}\right)M^2_f\left(\frac{h}{n} \right)^2,\\
    \mathrm{Var}(\mathrm{H}^{*}(f))&=O(1/n^2).
\end{align*}
\end{proposition}
\begin{proof}
See Appendix A.2.   
\end{proof}

Unlike the oracle esimator, the empirical approach relies on clustering, which deteriorates for smaller groups~\cite{choi2012stochastic}. Proposition~\ref{prop-oracle} shows that reducing the bandwidth decreases estimation bias, but the bandwidth must not drop below the threshold for effective clustering. The choice of blocks affects the graph structure: fitting one block simplifies it to an Erd\H{o}s-R\'enyi graph, while fitting as many blocks as nodes is impractical. Thus, the automatic bandwidth selection method discussed by~\cite{Olhede_2014} remains applicable for graphon entropy estimation.
\section{Graphon Entropy Estimators Under Special Assumptions}\label{sec:est-random-graphs}
In this section, we introduce tailored graphon entropy estimators and analyze their theoretical properties under specific conditions where additional structural information is available or certain assumptions hold. These scenarios, which highlight the trade-off between bias and variance, include cases where the graphon is constant, separable, or block-constant.
\subsection{Estimating the Entropy of a Constant Graphon}
For the Erdős–Rényi (ER) model, each potential edge between a pair of vertices is included independently with a fixed probability $p$ (where 
$0<p<1$), making it the simplest and one of the most well-studied random graph models. This setting is discussed in Definition~\ref{def-er}.
The entropy of this model then takes the following form:
\begin{equation}
    \label{eqn:binary2}
    {\mathrm{H}}(f^c)=\mathrm{H}(p)=-p\log(p)-(1-p)\log(1-p).
\end{equation}
As noted in Definition~\ref{def-entropy-hatami}, $\mathrm{H}(0) = \mathrm{H}(1) = 0$ to ensure continuity. To estimate this quantity, we compute $\widehat{p}$—the proportion of observed edges—based on the edge variables ${A_{ij}}_{i<j}$. The estimator is given by:
\begin{equation}
\widehat{p}=\frac{\sum_{i<j}A_{ij}}{\binom{n}{2}}.
\end{equation}
We note directly the first two moments from $\widehat{p}$ to be as follows: 
\begin{align}
\mathbb{E}\{  \widehat{p}\}&=p,\quad
\mathrm{Var}\{ \widehat{p} \}=\frac{p(1-p)}{\binom{n}{2}}
.\end{align}
The estimator of the constant graphon entropy is then given as:
\begin{equation}
    \label{eqn:binary3}
    \widehat{\mathrm{H}}(p)=- \widehat{p}\log( \widehat{p})-(1-\widehat{p})\log(1-\widehat{p}).
\end{equation}
Since the $A_{ij}$ are i.i.d. under this model, the Central Limit Theorem can be applied to $\sqrt{n^2}(\widehat{p} - p)$ to establish the asymptotic normality of $\widehat{p}$. Using the Delta method, we derive the asymptotic normality of  
$\sqrt{n^2}( \widehat{\mathrm{H}}(p)-{\mathrm{H}}(p))$. Specifically, this approach yields the following result:
\begin{equation}\label{er-normality}
  \!\!  \! \! \sqrt{ n^2}\{ \widehat{\mathrm{H}}(p)-{\mathrm{H}}(p)\}\overset{{\mathcal{L} }}{\rightarrow} \mathrm{N}\left(0,
   2p(1-p)(\log((1-p)/p))^2\right),
    \end{equation}
which gives the limiting distribution of the entropy for the constant edge probability random graph. 
\subsection{Estimating the Entropy of a Separable Graphon}
As mentioned in Definition~\ref{def-cm}, the separable graphon corresponds to the limit of graph sequences generated from the Chung--Lu model, and now we define the entropy of such graphons.
\begin{definition}[Entropy of a Separable Graphon]\label{def-cm-entropy} Let $f(x,y)=\rho_n g(x)g(y)$ for $0<g(x),g(y)<1$ and $0<x,y<1$ denote a separable graphon. Then, the entropy of the separable graphon is given by:
\begin{equation*}
    \mathrm{H}(f^s)=-\iint_{[0,1]^2}\{\rho_n g(x)g(y)\log(\rho_ng(x)g(y))+(1-\rho_ng(x)g(y))\log(1-\rho_ng(x)g(y))\}  \mathrm{d}x \mathrm{d}y.
\end{equation*}
\end{definition}
 To estimate the separable graphon from data, we define the estimator of $g(\cdot)$ to be as follows:
\begin{equation}
\widehat{g}(\xi_i)= \frac{d_i}{\sqrt{\|d\|_1}},  \quad i=1,\dots, n,     
\label{ghat2}
\end{equation}
where $d_i$ is the $i$th degree of our observed graph, and $\xi$ is the latent vector as detailed in Theorem~\ref{thm-aldous}. 
We note that this estimator has also been used parametrically by~\cite{franke2016network}. A useful property of this estimator is its simplicity, which facilitates easy computation.
Having defined the estimator for $g(\cdot)$ in Equation~\ref{ghat2}, we now construct the entropy estimator.
\begin{definition} Consider the setting of Definition~\ref{def-cm-entropy}, and let $\hat{g}(\cdot)$ be as given in Equation~\eqref{ghat2}. We define the entropy estimator of the separable graphon from the observed adjacency matrix $A\in \{0,1\}^{n\times n}$ as follows:
    \begin{align*}
   \widehat{ \mathrm{H}}(f^s)&=-\frac{1}{{n \choose 2}}\sum_{i<j}\{\hat \rho_n \hat {g}(\xi_i)\hat{g}(\xi_j)\log(\hat \rho_n \hat{g}(\xi_i)\hat{g}(\xi_j))+(1-\hat{\rho}_n \hat{g}(\xi_i)\hat{g}(\xi_j))\log(1-\hat{\rho}_n \hat{g}(\xi_i)\hat{g}(\xi_j))\}.
\end{align*}
\end{definition}
We now derive the asymptotic distribution of this estimator. To begin, note that:
\begin{equation}
\frac{\widehat{\rho}_n}{{\rho}_n} \overset{P}{\rightarrow}   1,
\end{equation}
and so we can use Slutsky's theorem to deal with the remaining random variation. 
\begin{theorem}
\label{ConfThm} Let $G(n,f^s)$, be generated via Definition~\ref{def-cm} sampled from $f^s(x,y)= \rho_n g(x) g(y)$ and let $0<\rho_n \leq 1$. As $n$ grows sufficiently large, then the distribution of $\widehat{\mathrm{H}}(f^s)$ is given as follows:
\begin{equation}
    \sqrt{\frac{n}{{\sigma^*}^2_{\widehat{\mathrm{H}}}}}\left\{\widehat{\mathrm{H}}(f^s)-{\mathrm{H}}(f^s)\right\} \overset{\mathcal{L}}{\to}  \mathrm{N}\left(0,1\right),
\end{equation}
where
${\sigma^*}^2_{\widehat{\mathrm{H}}}$ is given in Appendix B.
\end{theorem}
\begin{proof}
See Appendix B.
\end{proof}
The proof of Theorem~\ref{ConfThm} employs U-statistics, demonstrating their utility in addressing inference problems on random graphs. Unlike the case of a constant graphon, where observations can be treated as i.i.d. and the CLT applies directly, the dependence structure in graphs generated from a separable graphon necessitates alternative techniques. The theory of $U$-statistics provides a natural framework for addressing this dependence.  
\subsection{Estimating the Entropy of a Block--Constant Graphons}
The stochastic block model has served as a benchmark model for generating graphs with communities and community detection in graphs. We note that the block--constant graphon corresponds to the limit of a sequence of graphs generated by the stochastic block model.  
For notational simplicity, we first present the entropy of $W^b$. Since $f^b=\rho_n W^b$, its entropy follows directly, and we subsequently establish the convergence properties of $\hat{\mathrm{H}}(f^b)$.
\begin{definition}[Entropy of the block-constant graphon] 
Let $W^b$ denote the block-constant graphon as given in Definition~\ref{def-sbm}. Then, the entropy of the block-constant graphon is given by: 
\begin{equation}\label{ent-sbm}
    \mathrm{H}(W^b)=-\iint_{[0,1]^2}\Bigg\{\sum_{a,b} [\theta_{ab}\log(\theta_{ab})+(1-\theta_{ab})\log(1-\theta_{ab})]\mathds{1}(x \in I_a) \mathds{1}(y \in I_b) 
\,\mathrm{d}x\mathrm{d}y \Bigg\},
\end{equation}
where $I_a$ for $a \in [k]$ is a subinterval of the unit interval as detailed in Definition~\ref{def-sbm}.
\end{definition}

Simplifying Equation~\eqref{ent-sbm} by considering $k$ equal subintervals wlog, we obtain the following expression
\begin{equation}\label{entropy-sbm}
    \mathrm{H}(W^b)=-\sum_{a=1}^k \sum_{b=1}^k\frac{1}{k^2}\left\{\theta_{ab}\log(\theta_{ab})+(1-\theta_{ab})\log(1-\theta_{ab}) \right\}.
\end{equation} Appendix C.1 shows the derivation of~\eqref{entropy-sbm} from~\eqref{ent-sbm}. 

\begin{definition}[Entropy estimator of the block-constant graphon]\label{def:blockmodel-entropy-estimator}Let $W^b$ denote the block--constant graphon and $\hat{\theta}_{ab}$ denote a consistent estimator (least squares or maximum likelihood) of the connection probability between blocks $a$ and $b$ for $1\leq a,b \leq k$, where $k$ represents the number of blocks. Then, the entropy estimator of the block--constant graphon takes the following form:
\begin{equation}
\widehat{\mathrm{H}}(W^b)=-\sum_{a,b}\frac{1}{k^2}\left\{\hat{\theta}_{ab}\log(\hat{\theta}_{ab})+(1-\hat{\theta}_{ab})\log(1-\hat{\theta}_{ab}) \right\}.\end{equation}
    
\end{definition}

\begin{cor}\label{cor-entropy-sbm} Let $f^b=\rho_n W^b$, $\mathrm{H}({f^b})$ denote the entropy of $f^b$, and $\widehat{\mathrm{H}}(f^b)$ denote the graphon entropy estimator as given in Definition~\ref{def:blockmodel-entropy-estimator}. Then for any constant $C'>0$, there is a constant $C>0$ (where $C_1$ is a function of $C$), such that the following inequality
    \begin{align*}
    \small
    \left|\widehat{\mathrm{H}}(f^b)-\mathrm{H}(f^b)\right|
    &\leq C_1\left(\rho_n\frac{k^2}{n^2}+\rho_n\frac{\log k}{n}+\rho_n^2\sqrt{\frac{k}{n}}\right)^{1/2}+\mathcal{O}\left(\rho_n\frac{k^2}{n^2}+\rho_n\frac{\log k}{n}+\rho_n^2\sqrt{\frac{k}{n}}\right),
    \end{align*}
    holds with probability at least $1-\exp(-C'n)$.
\end{cor}
See Appendix C.2 for both the estimation method and the proof of Corollary~\ref{cor-entropy-sbm}.

\section{Simulation Study}\label{Sec:SimSt}
This section is organized into three subsections, each addressing key topics. We begin by evaluating the performance of the nonparametric (smooth) graphon entropy estimator under various conditions. Specifically, we analyze its behavior as the number of nodes increases, considering different sparsity regimes and various graphon models, using root mean squared error (RMSE) as the performance metric. Next, we show that the tailored graphon entropy estimator for separable graphons offers a slight performance improvement over the nonparametric estimator when the graphon is separable. However, for non-separable graphons, the nonparametric estimator yields superior results. Accordingly, we recommend using the nonparametric estimator for estimating the entropy of any exchangeable graph. Finally, we illustrate how the nonparametric graphon entropy estimator captures topological properties of graphs through examples of synthetically generated graphs, and compare its performance to traditional network entropy measures.
\subsection{Performance of the Graphon Entropy Estimator}
To examine the performance of the nonparametric graphon entropy estimator, we will analyze several scaled graphon examples that differ in nature. Namely, $f_1$, is a separable graphon, while $f_2$ is non-separable. Specifically, for $0<x, y<1$, constants $a_0, a_1, a_2$ lying on a simplex, and the regularized incomplete beta function $I_x(\alpha, \beta)$, where $\alpha_1, \alpha_2 \in(0,1]$, we consider the following graphons:
\begin{align}
\label{sim-f1}
f_1(x,y)&=\rho_n4xy,\\
\label{sim-f2}
f_2(x,y)&=\rho_n\left(a_0+4a_1 I_{x}^{-1}(\alpha_1,\alpha_1)I_{y}^{-1}(\alpha_1,\alpha_1)+4a_2 I_{1-x}^{-1}(\alpha_2,\alpha_2)I_{1-y}^{-1}(\alpha_2,\alpha_2)\right).
\end{align}
We restrict the choice for the shape parameters $\alpha_1, \alpha_2$ within the range $(0,1]^2$ to ensure Lipschitz continuity of the graphon $f_2$. By setting both $\alpha_1$ and $\alpha_2$ to 0.55, we introduce moderate variability in the degree distribution, allowing for a meaningful comparison with the separable graphon. This choice avoids the high variability that could result if the parameters were closer to the bounds of the range, allowing for more realistic modeling of interactions that reflect the properties of real-world networks. These selected values also prevent the graph from becoming too dense, a typical consequence of lower parameter values, thereby ensuring a trade-off that allows the graph to more accurately represent complex network interactions while maintaining reasonable density and variability.

For $\rho_n$ scaling with $n$, the graphon entropy does not converge to a single large sample limit. Therefore, we study the relative root mean square error (rRMSE) between $\widehat{\mathrm{H}}(f)$ and $\mathrm{H}(f)$, as this will also scale with $n$. We will define:
\[\mathrm{rRMSE}=\frac{\mathrm{RMSE}}{\mathrm{H}}=\frac{\sqrt{\frac{1}{N}\sum_{i=1}^N(\widehat{\mathrm{H}}_i -\mathrm{H})^2}}{\mathrm{H}}, \quad \mathrm{H}\neq 0,\]
where the $\widehat{\mathrm{H}}_i$ is computed for the $i$th graph realization out of $N>1$, and each $\widehat{\mathrm{H}}$ implicitly depends on $n$, unless $\rho_n$ is constant when scaling is not needed. This metric in a sense, enables us to say that $\widehat{\mathrm{H}}\approx {\mathrm{H}}$ unless the random variation around this value is too large. The simulation study confirms that after normalizing the RMSE with the true entropy at each $n$, the resulting rRMSE is decreasing with $n$, so the estimated entropy is getting \textit{close} to the true entropy.
Studying the behavior of the entropy estimator over a range of number of nodes, we plot the root mean squared error scaled by the true entropy, as this illustrates the large sample behavior of the graphon entropy estimator, see subfigures (b) and (d) of Figure~\ref{fig:RMSSE-entropy}. Normally, when estimating the graphon for the entire unit square, methods work well as long as $\rho_n=\Omega(\log(n)^{2+\epsilon}/n),$ for $\epsilon>0$. Here, we are not estimating a function as the primary quantity of interest, but a single scalar of the entropy, and so the last averaging over the $n^2$ symmetric edge variables exhibits a better convergence rate with reasonable performance even for very sparse graphs.

To see how graphon entropy estimator responds to sparsity, but while the sparsity parameter $\rho_n$ does not change with $n$, we consider various constant density regimes, decreasing in order, and for this we use the usual root mean square error, since now the true entropy remains unchanged, and we consider:
\begin{equation*}
    \mathrm{RMSE}=\sqrt{\frac{1}{N}\sum_{i=1}^N(\mathrm{H}_i -\mathrm{H})^2}.
\end{equation*}
We observe from subfigures (a) and (c) of Figure~\ref{fig:RMSSE-entropy} that as the constant $\rho_n$ decreases, the RMSE exhibits a similar downward trend, indicating improved estimates with increasing node size.

\begin{figure}[ht!]
    \centering
    \begin{minipage}{0.48\linewidth}
        \centering
        \includegraphics[width=\linewidth]{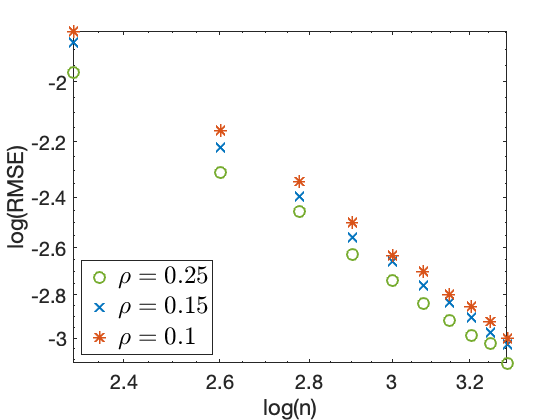}
        \caption*{(a)}
    \end{minipage}
    \hfill 
    \begin{minipage}{0.48\linewidth}
        \centering
        \includegraphics[width=\linewidth]{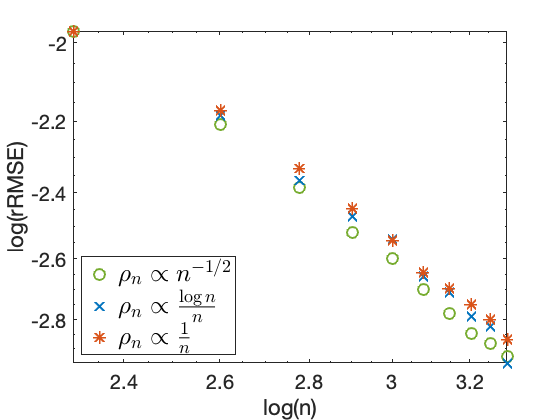}
        \caption*{(b)}
    \end{minipage}
       \begin{minipage}{0.48\linewidth}
        \centering
        \includegraphics[width=\linewidth]{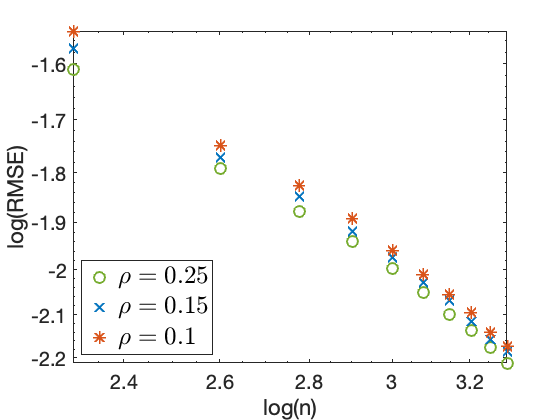}
        \caption*{(c)}
    \end{minipage}
    \hfill 
    \begin{minipage}{0.48\linewidth}
        \centering
        \includegraphics[width=\linewidth]{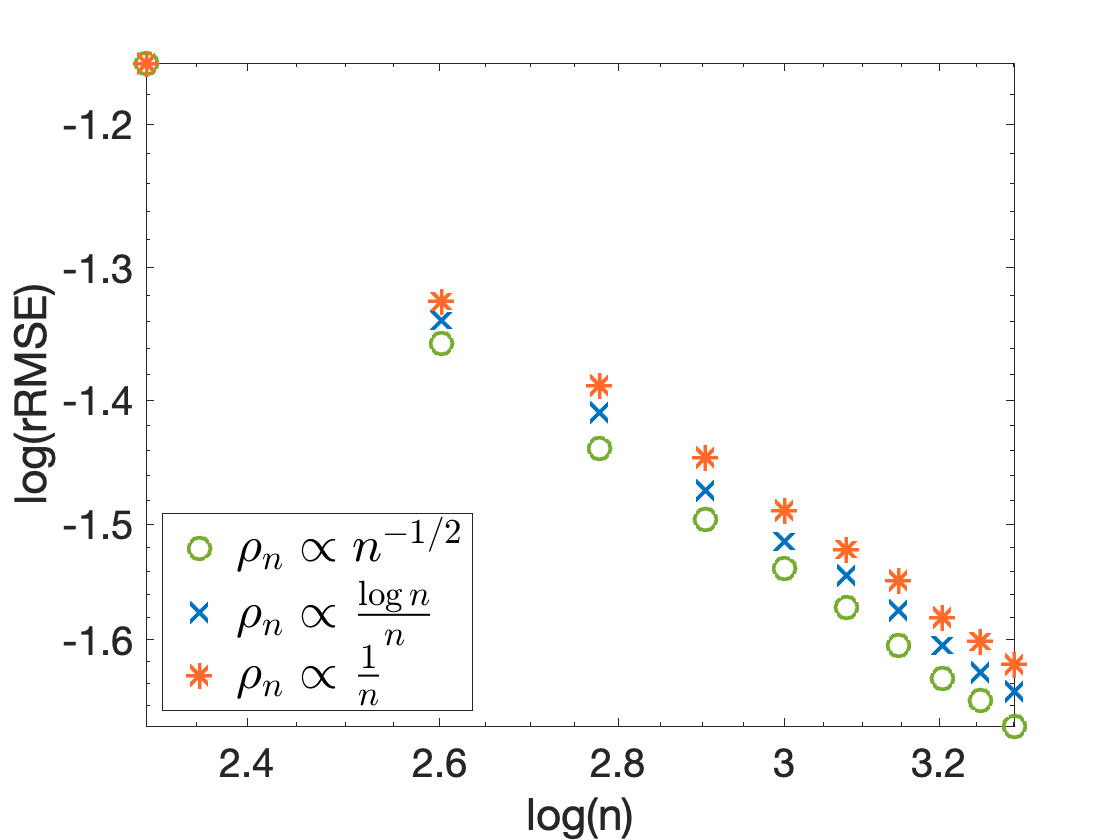}
        \caption*{(d)}
    \end{minipage}
 \caption{Log-log plots of RMSE and rRMSE decay for the graphon entropy estimator across node sizes from 200 to 2000, averaged over 100 trials. Subfigures: (a) RMSE of the estimator of $\mathrm{H}(f_1)$ over various constant $\rho$ for $f_1$ as given in~\eqref{sim-f1}, (b) rRMSE of the estimator of $\mathrm{H}(f_1)$ over various sparsity regimes $\rho_n$, (c) RMSE of the estimator of $\mathrm{H}(f_2)$ over various constant $\rho$, for $f_2$ as given in~\eqref{sim-f2} (d) rRMSE of the estimator of $\mathrm{H}(f_2)$ over various sparsity regimes $\rho_n$.}\label{fig:RMSSE-entropy}
\end{figure}

\subsection{Comparison of the Introduced Graphon Entropy Estimators}

We now refer to Table~\ref{table1} to compare the performance of the proposed graphon entropy estimators in simulated networks, where `$\widehat{\mathrm{H}}$' stands for graphon entropy estimator.
For the separable graphon ($f_1$) as given in Equation~\eqref{sim-f1}, the separable graphon entropy estimator performs best (as expected), as the nonparametric (smooth) graphon entropy estimator allows for a more flexible approximation, but exhibits finite sample size bias. The variances are comparable between the separable graphon entropy estimator and the smooth graphon entropy estimator, but the bias thus beats out the more general smooth graphon entropy estimator. The constant graphon entropy estimator performs best on variance but is, of course, (in general) biased.
\begin{table}[!htb]
\centering
\begin{tabular}{@{}llll@{}}
\toprule
 & Constant $\widehat{\mathrm{H}}$       & Separable $\widehat{\mathrm{H}}$ & Nonparametric $\widehat{\mathrm{H}}$  \\ \midrule
MSE         & $1.85\times 10^{-2}$ & $1.98\times 10^{-4}$ & $2.31\times 10^{-4}$ \\
            & $(3.75\times 10^{-4})$ & $(2.88\times 10^{-5})$ & $(3.07\times 10^{-5})$ \\
Bias Squared & $1.83\times 10^{-2}$ & $6.91\times 10^{-6}$ & $6.25\times 10^{-5}$ \\
           
Variance     & $1.97\times 10^{-4}$ & $1.93\times 10^{-4}$ & $1.7\times 10^{-4}$ \\
           
\bottomrule
\end{tabular}
\caption{Comparison of the proposed graphon entropy estimators for the separable graphon $f_1$ (as defined in~\eqref{sim-f1}), with $n=512$ and $\rho_n = 0.25$, over 100 repetitions. Standard errors are shown in parentheses.}\label{table1} 
\end{table}

\begin{table}[!htb]
\centering
\begin{tabular}{@{}llll@{}}
\toprule
             & Constant $\widehat{\mathrm{H}}$          & Separable $\widehat{\mathrm{H}}$ & Nonparametric $\widehat{\mathrm{H}}$      \\ \midrule
MSE          & $1.07 \times 10^{-3}$             & $1.04 \times 10^{-3}$     & $1.06\times 10^{-5}$ \\
             & $(1.99 \times 10^{-6})$           & $(2.33 \times 10^{-6})$   & $(1.18 \times 10^{-6})$ \\
Bias Squared & $1.07\times 10^{-3}$              & $1.04 \times 10^{-3}$     & $5.09 \times 10^{-6}$ \\
             
Variance     & $9.24 \times 10^{-8}$             & $1.3 \times 10^{-7}$      & $5.6 \times 10^{-6}$ \\
             
\bottomrule
\end{tabular}
\caption{Comparison of the proposed graphon entropy estimators for the non-separable graphon $f_2$ (as defined in~\eqref{sim-f2}), with parameters $a_0=0.8$, $a_1=0.5$, $n=512$, and $\rho_n = 0.25$, over 100 repetitions. Standard errors are shown in parentheses.}\label{table2}
\end{table}

On the other hand, $f_2$ as given in Equation~\eqref{sim-f2}, is not a separable graphon, so the graphon entropy estimator based on assuming separability of the graphon is no longer the best-performing method. The nonparametric graphon entropy estimator, in this case, has a smaller bias and a somewhat larger variance, but is an improvement on the other methods, leaving the general graphon entropy estimator applicable to any exchangeable graph.

\subsection{Graphon Entropy Estimator and Graph Structure}
\subsubsection{Graphon Entropy Estimator and Degree Distribution}
We investigate the changes in the nonparametric graphon entropy estimator associated with different degree structures while maintaining the same sparsity parameter, $\rho_n$. We also compare our results with existing network entropy measures. To do so, we use a separable $f_1$ and constant graphon $f_2$ to resemble the Chung-Lu and the Erd\H{o}s-R\'enyi models, respectively, given as below:
$$f_1(x,y)=\rho_n4xy,$$ is separable, thereby degree-based, and the second graphon 
$$f_2=p,$$ for $p=0.252$ is constant and chosen to maintain the same edge density as the first graph.

The graph generated by the first graphon exhibits high variability in its degree structure, reflecting degree heterogeneity. In contrast, the graph from the second graphon, which resembles an Erdős–Rényi (ER) graph, exhibits a more uniform degree distribution. While both graphs have the same expected edge density, the structured nature of the first graph constrains edge placement, leading to lower entropy. In contrast, the second graph, which lacks such structural constraints, exhibits higher entropy, aligning with its more random nature. The estimated entropy of the first graphon equals $\widehat{\mathrm{H}}(f_1)=0.4296$ with a $95\%$ confidence interval of $[0.4172,0.4484]$, averaged over $100$ runs. On the other hand, the estimated graphon entropy of the second one equals $\widehat{\mathrm{H}}(f_2)=0.5654$ with a $95\%$ confidence interval of $[0.5432,0.5897]$, averaged over $100$ trials.

These results demonstrate that the graphon entropy estimator captures information related to the degree distribution. This is illustrated by keeping the scaling or sparsity parameter, $\rho_n$, constant across both graphs.
\subsection*{Graphon Entropy Estimator Compared to Traditional Network Entropies}
We compare the nonparametric graphon entropy estimator with traditional network entropy measures, namely the degree sequence Shannon entropy, the random walker Shannon entropy~\cite{small2013complex}, the Random Walker Kolmogorov Sinai entropy (KS entropy)~\cite{demetrius2005robustness}, and the von Neumann entropy~\cite{braunstein2006laplacian} presenting the results in Table~\ref{tab:entropy_comparison}. As expected, the degree sequence entropy is higher for the degree-based graph than for the Erdős–Rényi (ER) graph, and a similar trend is observed for the KS entropy. In contrast, the random-walker entropy is slightly higher for the ER graph, though the difference is less pronounced than that observed with the graphon entropy estimator. Von Neumann entropy remains high for both models, showing limited sensitivity to structural differences and primarily reflecting overall connectivity. These results illustrate that while traditional entropy measures are feature-dependent, each captures a different structural aspect of the graph. The graphon entropy estimator, however, provides a broader perspective by quantifying the randomness or disorder inherent in the graph-generating mechanism, effectively distinguishing an Erdős–Rényi (ER) graph from a Chung–Lu graph.

\begin{figure}[ht!]
    \centering
    \begin{minipage}{0.48\linewidth}
        \centering
        \includegraphics[width=\linewidth]{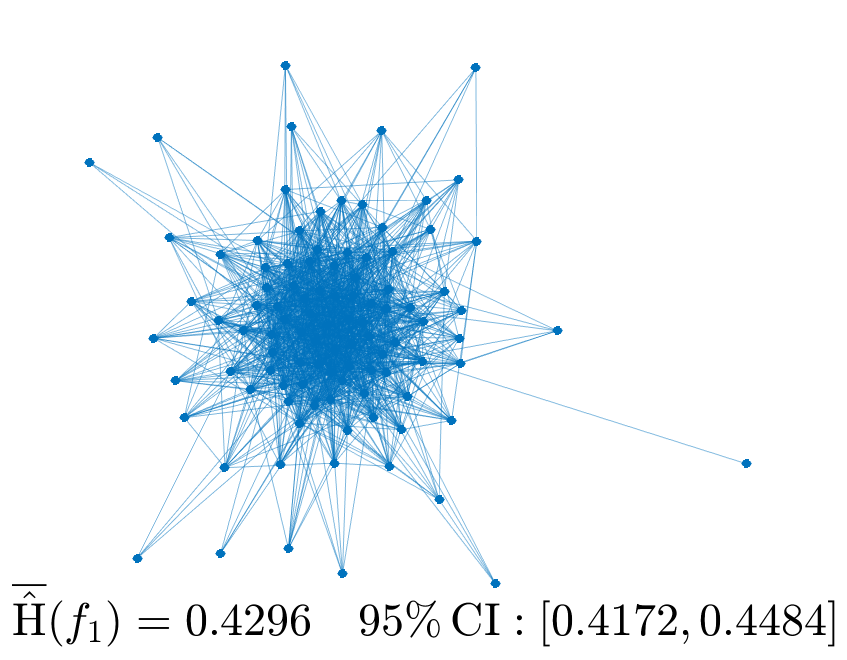}
        \label{fig:tc-sbm-percolated}
    \end{minipage}
    \hfill 
    \begin{minipage}{0.48\linewidth}
        \centering
        \includegraphics[width=\linewidth]{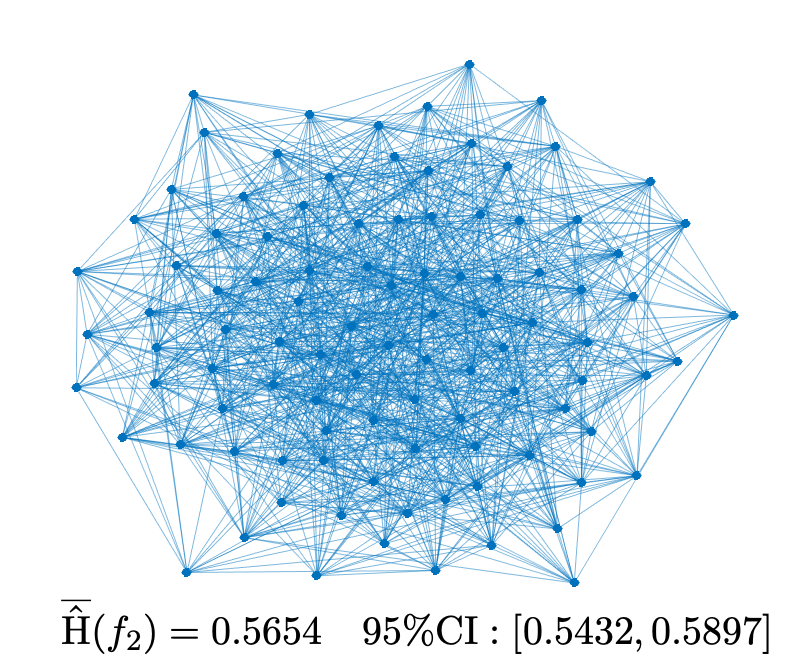}
        \label{fig:tc-sbm-different-thetas}
    \end{minipage}
    \caption{Graphs generated on $n=500$ nodes with $
\rho_n=0.25$ and respective graphons, $f_1$ as given in~\eqref{sim-f1} and $f_2$ given in~\eqref{sim-f2}.}
\end{figure}

\begin{table}[h]
    \centering
    \begin{tabular}{l c c}
        \toprule
        Entropy Measure & Chung–Lu $(\rho_n=0.25)$ & {Erd\H{o}s–Rényi} $(p=0.252)$ \\ 
        \midrule
        Graphon entropy estimator & 0.4296 & 0.5654 \\ 
        Degree sequence entropy & 0.8448 & 0.5822 \\ 
        Random-Walker entropy & 0.7342 & 0.7796 \\ 
        KS entropy & 0.8276 & 0.7809 \\ 
        VN entropy & 0.9689 & 0.9988 \\ 
        \bottomrule
    \end{tabular}
    \caption{Comparison of the graphon entropy estimator with network entropy measures for separable and constant graphons.}
    \label{tab:entropy_comparison}
\end{table}

\subsubsection{Graphon Entropy Estimator and Community Structure}
We investigate whether the nonparametric graphon entropy estimator effectively captures information about community structure, an important feature of graphs. Figure~\ref{SBM-entropy} illustrates a series of graphs with the same set of nodes and identical sparsity regimes but varying community structures. The first graph exhibits an emphasized community structure, while subsequent graphs progressively lose this clarity, with the final graph resembling an Erd\H{o}s–R\'enyi (ER) graph. We note that the graphon entropy estimator measures the level of randomness or uniformity, which could be phrased as structural disorganization rather than organization. The estimated average graphon entropy for each graph, shown below their respective images in Figure~\ref{SBM-entropy}, is averaged over 100 runs. We observe that graphs with pronounced community structures have lower entropy compared to those with more diffuse structures. This finding aligns with the concept of entropy, which associates higher entropy with greater uniformity.
\begin{figure}[ht!]
    \centering
    \begin{minipage}{\textwidth}
        \centering
        \begin{minipage}{0.31\linewidth}
            \centering
            \includegraphics[width=\linewidth]{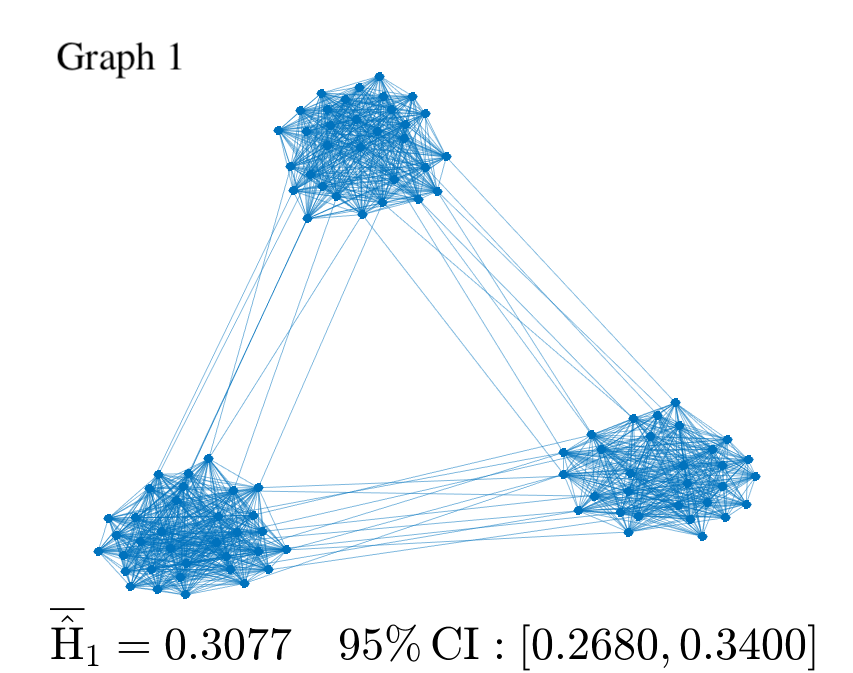}
        \end{minipage}
        \hfill
        \begin{minipage}{0.31\linewidth}
            \centering
            \includegraphics[width=\linewidth]{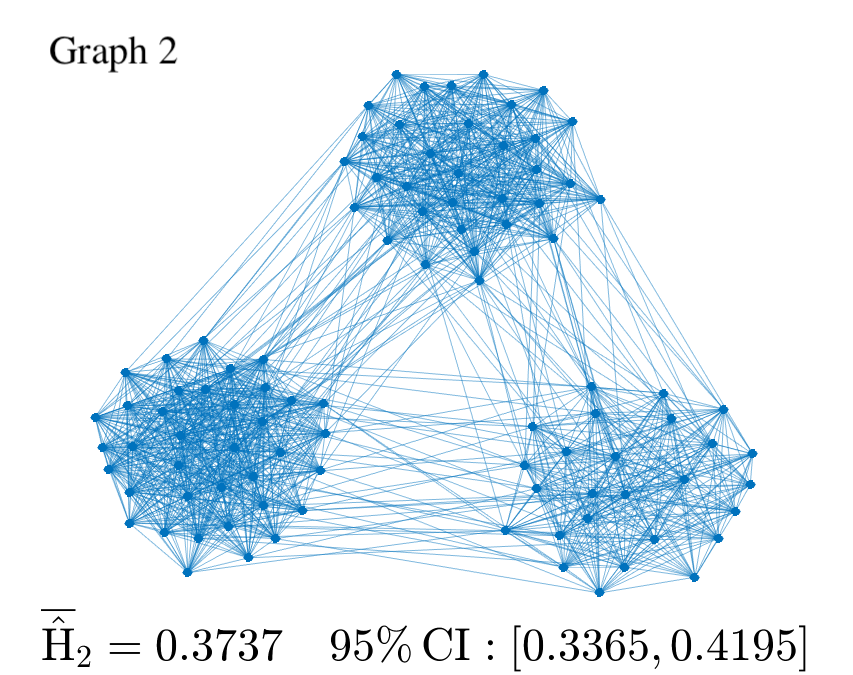}
        \end{minipage}
        \hfill
        \begin{minipage}{0.31\linewidth}
            \centering
            \includegraphics[width=\linewidth]{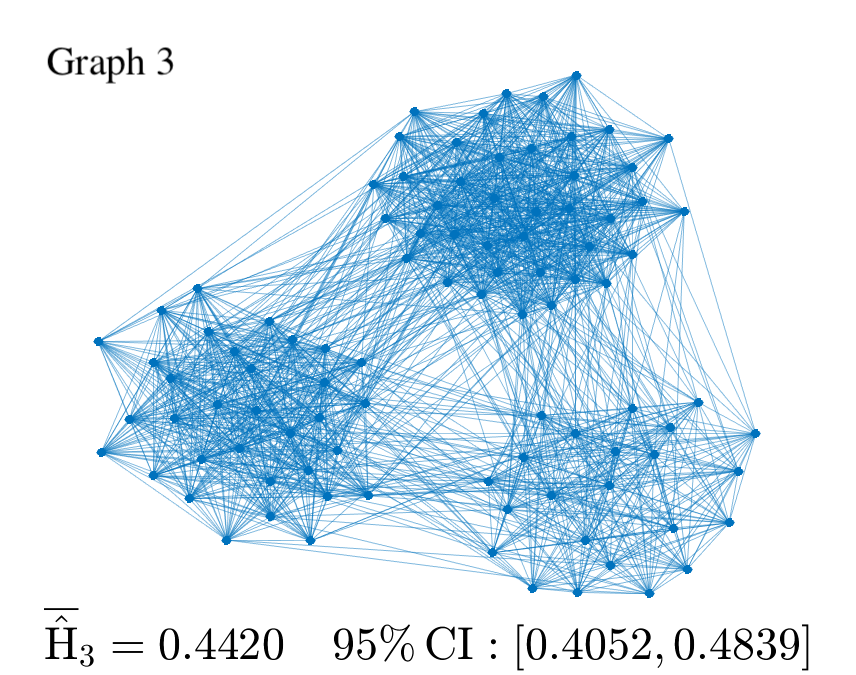}
        \end{minipage}
        
        \vspace{0.5cm}

        \begin{minipage}{0.31\linewidth}
            \centering
            \includegraphics[width=\linewidth]{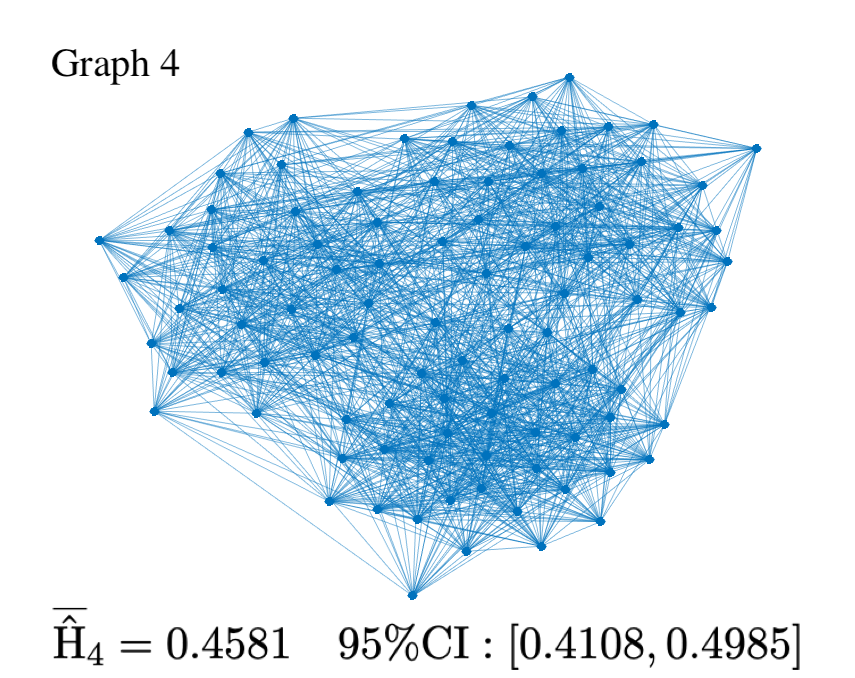}
        \end{minipage}
        \hfill
        \begin{minipage}{0.31\linewidth}
            \centering
            \includegraphics[width=\linewidth]{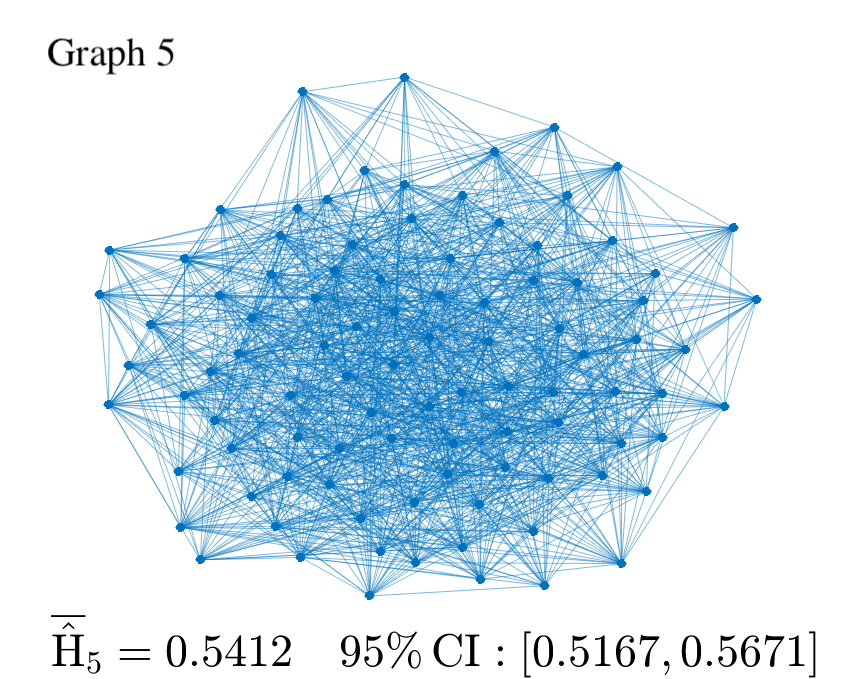}
        \end{minipage}
        \hfill
        \begin{minipage}{0.31\linewidth}
            \centering
            \includegraphics[width=\linewidth]{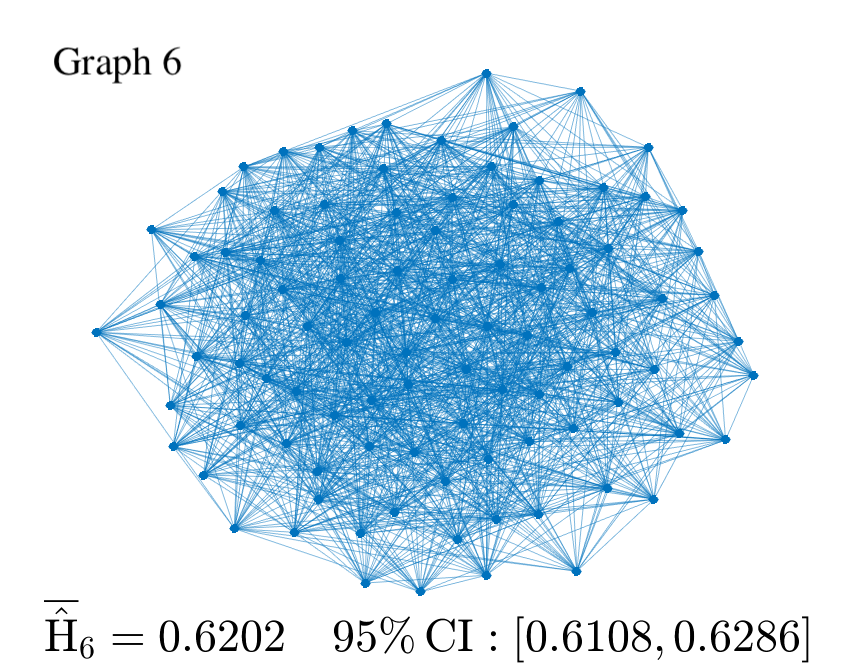}
        \end{minipage}

        \caption{Graphs generated on $n=96$ nodes with $\rho_n=0.5$ and varying community structures.}
        \label{SBM-entropy}
    \end{minipage}

    \vspace{1cm} 
\end{figure}

\subsection*{Graphon Entropy Estimator Compared to Traditional Network Shannon Entropies}

Figure~\ref{fig:entropy-comparison} presents a comparative analysis of entropy measures across graphs with varying community structures depiced in Figure~\ref{SBM-entropy}. Kolmogorov-Sinai (KS) entropy exhibits a slight downward trend, with the highest entropy for the most structured graph and the lowest for the Erdős–Rényi (ER) graph. The random walker entropy remains nearly constant across all graphs. The degree-sequence entropy also remains nearly constant for all graph realizations. We have not plotted the von Neumann entropy to avoid distorting the scale of the axes, as it remains consistently high at 0.99, with only slight variations between 0.9903 and 0.9908. In contrast, the graphon entropy estimator exhibits a gradual increase, with entropy values rising as community structure becomes less distinct and reaches the highest value for the ER graph. This trend, observed in Figure~\ref{fig:entropy-comparison}, underscores the ability of the graphon entropy estimator to quantify randomness and uniformity, distinguishing structured graphs from those with diffuse connectivity.
\begin{figure}[!htb]
    \begin{minipage}{\textwidth}
        \centering
        \includegraphics[width=0.7\textwidth]{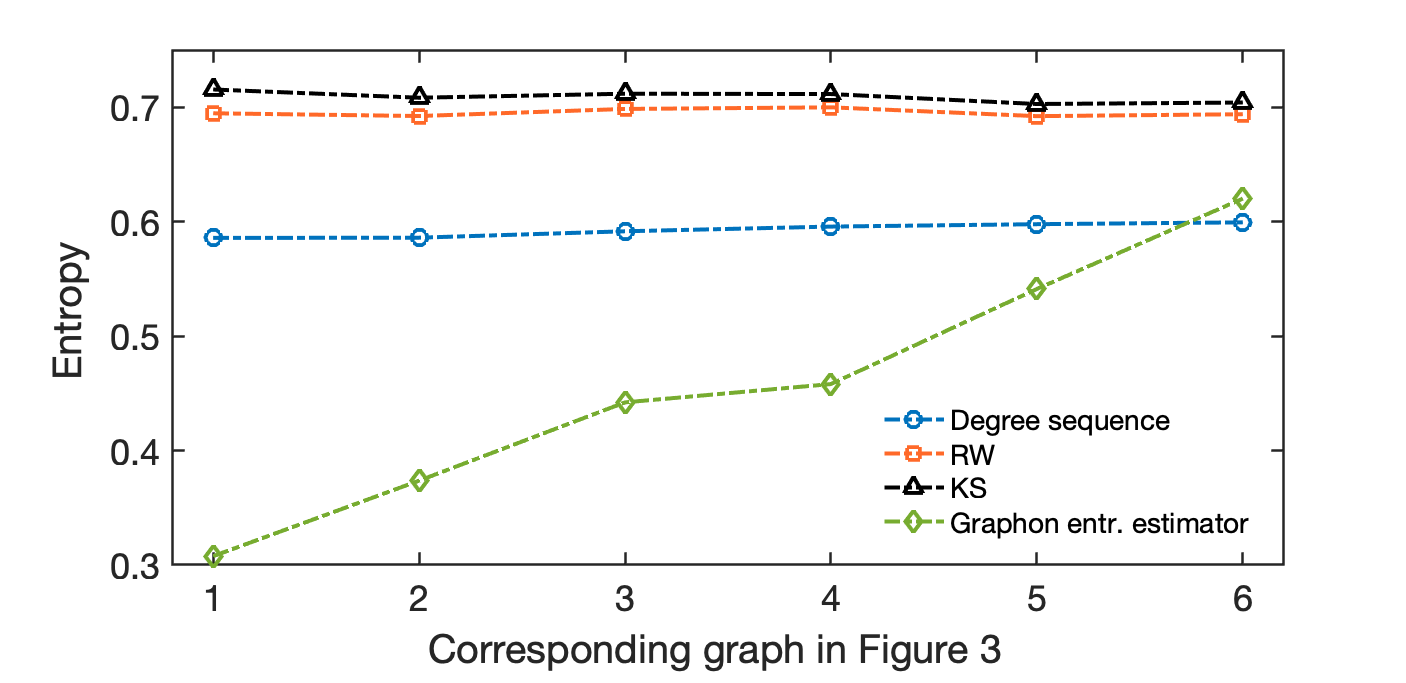}
       \caption{Comparison of the nonparametric graphon entropy estimator with some network entropy measures across graphs with varying community structures  as given in Figure~\ref{SBM-entropy}. }
        \label{fig:entropy-comparison}
    \end{minipage}
\end{figure}

\section{Real Data Example}\label{Sec:RealData}

We analyze a collection of graphs representing the behavior of corporate boards in Norway, with the data set and subfigure (a) of Figure~\ref{fig:boards-data} publicly available online (\url{https://graphs.skewed.de/net/board_directors}). This data set records the connections formed by individuals who served on the same boards of Norwegian public limited companies and was previously studied by~\cite{seierstad2011few}. The data were compiled from the Norwegian Business Register as of 5 August 2009 and were used to examine the impact of legislation on gender representation. The law, enacted in 2008, mandated at least $40\%$ representation of each gender on the boards of public limited companies~\cite{seierstad2011few}. According to~\cite{seierstad2011few}, the legislation successfully increased women's participation on boards, broadening the pool of qualified female directors.
\begin{figure}[ht!]
    \centering
    \begin{minipage}{0.48\linewidth}
        \centering
        \includegraphics[width=\linewidth]{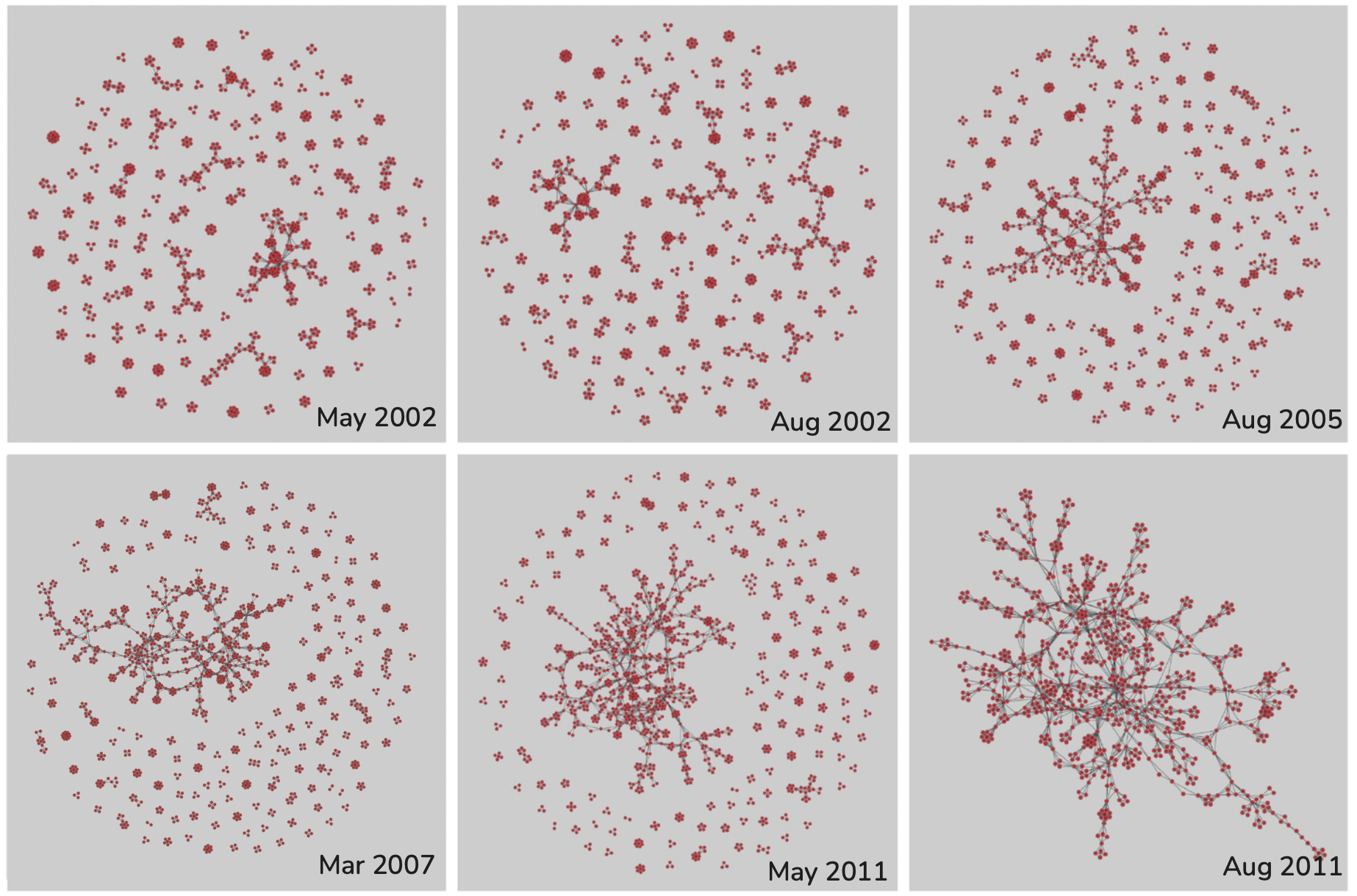}
        \caption*{(a)}
    \end{minipage}
    \hfill 
    \begin{minipage}{0.48\linewidth}
        \centering
        \includegraphics[width=1.15\linewidth]{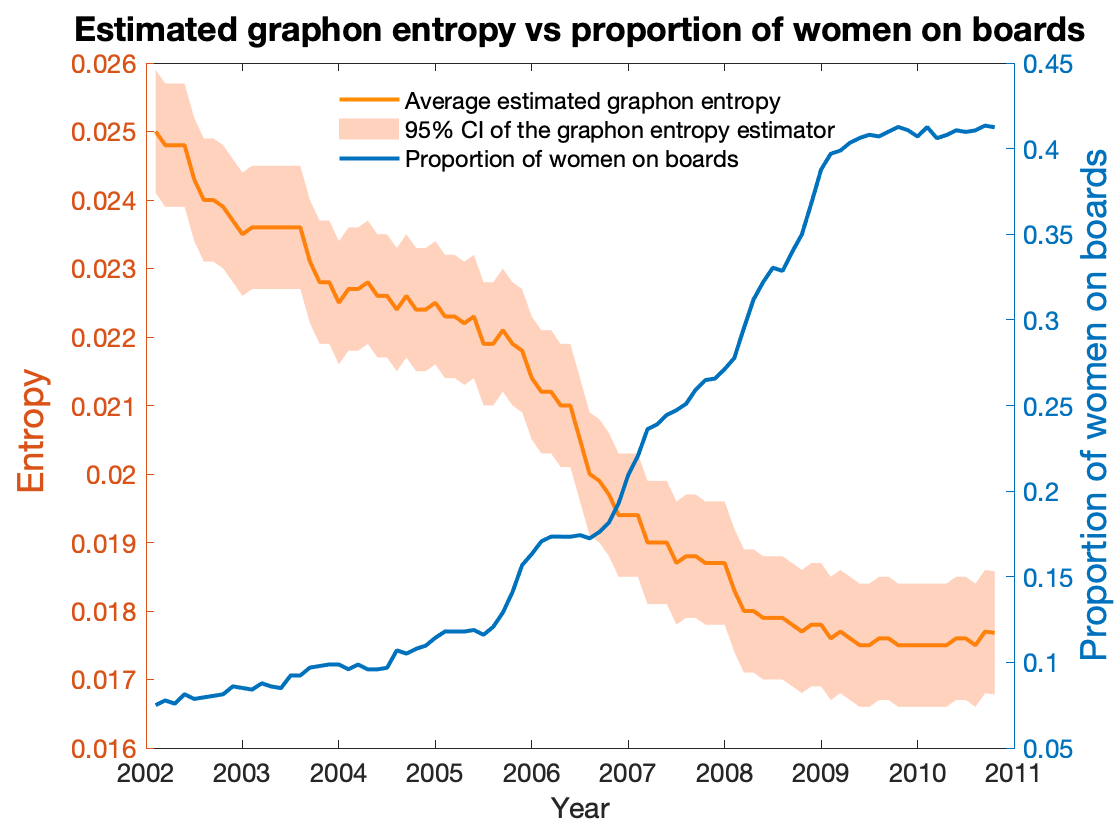}
        \caption*{(b)}
    \end{minipage}
    \caption{Subfigure (a) depicts selected networks of boards of directors in Norway’s public companies to illustrate the evolving network dynamics, and subfigure (b) shows the average proportion of women directors (blue--right) given by~\cite{seierstad2011few} and average graphon entropy change over time (orange--left).}\label{fig:boards-data}
    \end{figure}

We estimate the graphon entropy of each network observation, as detailed in Section~\ref{sec-estimation-arbitrary}, to examine the insights provided by entropy changes over time, given the temporal nature of the data set. A key challenge in network inference is quantifying uncertainty from a single observed network. We address this by using the bootstrapping method described in~\cite{green2022bootstrapping}. First, we estimate the graphon from the observed graph. Then, we generate new latent variables $\xi_i^*\sim U(0,1)$ for $i \in [n]$ at each step and generate new graphs from using the estimated graphon and the new latent variables $\xi^*_i$. Through this process, we generate a thousand simulated graphs, each representing a possible network realization. We estimate the graphon entropy for each of the realizations, generating a distribution from bootstrap samples. From this distribution, we calculate the average estimated graphon entropy and construct the $95\%$ confidence intervals to quantify the uncertainty of our estimates. These intervals (left: shaded light orange) are plotted alongside the average (left: orange curve). 

The average proportion of women on boards is shown with the blue curve (right) in subfigure (b) of Figure~\ref{fig:boards-data}. This ratio demonstrates a substantial increase in the representation of women on boards during the law enforcement period. The time from the announcement of the law's enforcement until it took effect, or between 2005 and 2008, had the biggest growth in particular. However, similarly, the orange curve (left) in Subfigure (b) of Figure~\ref{fig:boards-data} indicates that the highest decrease in average estimated entropy took place from mid-2005 to December 2008. The aligning change of the average estimated graphon entropy with the proportion of women on boards and consequently the changes occuring in the networks as given in subfigure (a) of Figure~\ref{fig:boards-data} shows that the nonparametric graphon entropy estimator is promising in capturing changes occurring in this graph collection. The reason for the decrease in entropy in parallel with the increase in the average proportion of women on boards could be attributed to the findings by~\cite{seierstad2011few}, who state that women have more social capital than men and that the average betweenness score of women is $70\%$ higher than that of men. Following this line of reasoning, the female director boards in the dataset tend to form connected components, which contribute to increasing the overall connectivity and structural coherence of the graph. As a result, the network becomes more predictable, leading to a reduction in graphon entropy. See Figure \ref{SBM-entropy} for an illustration of the relationship between graphon entropy and the evolving graph structure. This phenomenon is evident in subfigure (a) of Figure~\ref{fig:boards-data}, where we present selected snapshots of the graph's evolution, with the first and last graphs corresponding to the first and last time points, respectively. The rest are selectively chosen from the collection to illustrate the evolving dynamics. Refinement of a graph structure to achieve a connected graph is inherently associated with changes in edge density. We observe variability in the average degree over different time intervals, which aligns with our results. These snapshots begin with the first graph in the dataset and proceed through selected stages up to the final graph. The graphon entropy estimator, which takes into account factors like sparsity, degree distribution, and community structure of real-world graphs, adjusts accordingly. This adjustment is crucial as, unlike in simulation studies, the density parameter is not fixed across the graphs under study. 
The decline in entropy aligns with the increasing participation of women on corporate boards, leading to a more structured and predictable network. This supports findings that female directors play a central role in board connectivity, reducing randomness in network formation.

\section{Conclusion}\label{Sec:Disc} 
In conclusion, we introduced a suite of graphon entropy estimators as complexity measures for large graphs, effectively characterizing the graph-generating mechanism while remaining independent of graph representation. Building on a smoothness assumption, we used a block model approximation to achieve a tractable and consistent nonparametric graphon entropy estimator, applicable to any exchangeable graph. We further developed tailored estimators under additional assumptions, including constant, separable, and block-constant graphons, and established  Central Limit Theorems or convergence rates for each case. Simulations validated these findings, showing reduced error with increasing number of nodes, even in very sparse regimes. We also explored graphon entropy estimator's relationship with sparsity, degree distribution, and community structure, demonstrating its ability to capture meaningful information related to graphs. We further compared the nonparametric graphon entropy estimator with traditional network entropy measures. In real-world applications, the nonparametric graphon entropy estimator proved effective in analyzing time-varying graphs, aligning with prior observations~\cite{seierstad2011few}. These results underscore the estimator’s potential for addressing complex network inference problems and serving as a backbone for developing multivariate information-theoretic measures based on graphons.

\newpage
\bibliographystyle{acm}

\newpage
\section*{A.1 $\quad$ Proof of Theorem 2}\label{app-a.1}
\begin{proof}~We note that the Kullback-Liebler divergence between two Bernoulli distributions $p$ and $p+\delta$, where $0<p<1$ and $-p\leq \delta \leq 1-p$ is given as:
\begin{align}\label{KL-div}
D(p+\delta\|p)=(p+\delta) \log \left(\frac{p+\delta}{p}\right)+(1-(p+\delta)) \log \left(\frac{1-(p+\delta)}{1-p}\right).
\end{align}
We then bound the following expression:
\begin{align}
\nonumber \left|\widehat{\mathrm{H}}(f) - \mathrm{H}(f)\right| &= \left|\iint_{[0,1]^2} (f - \hat{f})\log(f) + [(1 - f) - (1 - \hat{f})]\log(1 - f) \, dx \, dy - D(\hat{f} \parallel f) \right| \\
\nonumber &\leq \left( \iint_{[0,1]^2} (f - \hat{f})^2 \, dx \, dy \right)^{1/2} \\
\nonumber &\quad \cdot \left\{ \left( \iint_{[0,1]^2} (\log(f))^2 \, dx \, dy \right)^{1/2} + \left( \iint_{[0,1]^2} (\log(1 - f))^2 \, dx \, dy \right)^{1/2} \right\} \\
\nonumber &\quad + |D(\hat{f} \parallel f)|\\
&\label{eq-bound-entropy}\leq C_1(\delta^2(\hat{f},f))^{1/2}+C_2\delta^2(\hat{f},f),
\end{align}
for some constants $C_1,C_2>0$. Note that following Lemma C.9 of~\cite{wolfe2013nonparametric}, we obtain:
\begin{equation*}
  D(\hat{f}||f)\leq C_{\mathrm{min}}\iint_{[0,1]^2} (\hat{f}-f)^2 \ \mathrm{d}x\mathrm{d}y.  
\end{equation*}
The dominating term in bounding the error between the graphon entropy estimator and the true graphon is the first term. We now refer to Corollary 3.6 in~\cite{klopp2017oracle} for the optimal rate of graphon estimation and employ Equation~\eqref{eq-bound-entropy} to obtain the following:
\begin{align*}
\small
    &\left|\widehat{\mathrm{H}}(f)-\mathrm{H}(f)\right| \\
    &\leq C\left(\rho_n^{-2\alpha/(\alpha+1)}n^{-2\alpha/(\alpha+1)}+\frac{\rho_n \log n}{n}+\frac{\rho_n^2}{n^{\alpha}}\right)^{1/2}+\mathcal{O}\left(\rho_n^{-2\alpha/(\alpha+1)}n^{-2\alpha/(\alpha+1)}+\frac{\rho_n \log n}{n}+\frac{\rho_n^2}{n^{\alpha}}\right),
\end{align*}
 with probability at least $1-\exp(-C'n)$, uniformly over $f \in \mathcal{F}_{\alpha}(M)$, where $C'>0$ is any constant and $C>0$ depend only on $M$ and $C'$.

\end{proof}

\section*{A.2 $\quad$ Proof of Proposition 1}\label{app-a.3}
Assume we generate realizations from~\eqref{sparse-graphon}. We shall consider the oracle estimator of Eqn. (3.2) based on the oracle graphon estimator with a fixed bandwidth $h$ of (3.2) for $f(x,y)$. For a given $h$ we can calculate the mean and variance of Eqn.(3.3), to determine how to balance the two. We first note that from~\cite{Olhede_2014}
\begin{align}
\nonumber
\mathbb{E} \{\bar{A}^{\ast}_{ab}(\tilde{z})\}&=\rho_n \bar{f}_{a,b}(1+o(1))\\
\left|\mathrm{Var}\{ \bar{A}^{\ast}_{ab}(\tilde{z})\}-\frac{\rho_n \bar{f}_{a,b}-\rho_n^2\overline{f^2}_{ab}}{h^2_{ab}}\right|&<\frac{\rho_n M_f}{h^2_{ab}(2n)^{1/2}}+
\frac{\rho_n^2 M^2_f}{2n}.
\end{align}
We can deduce from these two moments that as long as $\rho_n h_{ab}^2=\omega(1)$
\begin{align}
    \bar{A}^{\ast}_{ab}(\tilde{z})/\rho_n\overset{P}{\rightarrow}
\bar{f}_{ab}.
\end{align}
Secondly we note the following from Eqn. (3.2):
\begin{align}
\widehat{\mathrm{H}}^*(f)&=-\sum_{a,b}\frac{h_a h_b}{n^2}\left\{\bar{A}_{ab}^{\ast}(\widetilde{z})\log(\bar{A}_{ab}^{\ast}(\widetilde{z}))+(1-\bar{A}_{ab}^{\ast}(\widetilde{z}))\log(1-\bar{A}_{ab}^{\ast}(\widetilde{z})) \right\},
\end{align}
where $h_a$ is the number of nodes in group $a$ such that $\sum_{a=1}^k h_a=n$, or otherwise $h_a/n$ represents the length of one of the $k$ subintervals of $[0,1]$.
To understand the mean and variance of $\widehat{\mathrm{H}}^*(f)$ we shall Taylor--expand the expression in the variables $\bar{A}_{ab}^{\ast}(\tilde{z})$, and we find that (assuming uniform integrability of $\bar{A}$):
\begin{align}
  \left| \mathbb{E}\left\{  \widehat{\mathrm{H}}^*(f)\right\}-{\cal H}( f)\right|&\leq \left(\frac{1}{\rho_n \min_{a,b}\bar{f}_{ab}}+\frac{1}{1-\rho_n \max_{a,b}\bar{f}_{ab}}\right) M_f^2 \left(\frac{h}{n}\right)^2.
\end{align}
This in essence determines the bias of the entropy estimation, which is monotonically decreasing with $h$ decreasing. The expression also explicitly shows how difficult estimation near the boundaries of the parameter values is.

We can also calculate the variance of the oracle estimator $\widehat{\mathrm{H}}^*(f)$. We find that this variance takes the form after linearization via delta method
\begin{align}
\nonumber
\mathrm{Var}\{\widehat{\mathrm{H}}^*(f) \}&=\mathrm{Var}\{\sum_{a,b}\frac{h_a h_b}{n^2}\left\{\bar{A}_{ab}^{\ast}(\widetilde{z})\log(\bar{A}_{ab}^{\ast}(\widetilde{z}))+(1-\bar{A}_{ab}^{\ast}(\widetilde{z}))\log(1-\bar{A}_{ab}^{\ast}(\widetilde{z})) \right\} \}\\
&=\mathrm{Var}\{\sum_{a,b}\frac{h_a h_b}{n^2}\{2+\log(\rho_n \bar{f}_{ab}) +\log(1-\rho_n \bar{f}_{ab})\}\bar{A}^{\ast}_{ab}\}(1+o(1)).
\end{align}
As most quantities are positive in this sum, we may note
\begin{align}
\nonumber
\mathrm{Var}\{\widehat{\mathrm{H}}^*(f) \}&\leq 
\{2+|\log(\rho_n \|f\|_{\min})| +|\log(1-\rho_n \|f\|_0)|\}^2\mathrm{Var}\{\sum_{a,b}\frac{h_a h_b}{n^2}\bar{A}^{\ast}_{ab}\}(1+o(1))\\
\nonumber
&\leq 
\{2+|\log(\rho_n \|f\|_{\min})| +|\log(1-\rho_n \|f\|_0)|\}^2\frac{1}{n^4}\mathrm{Var}\{2\sum_{i<j}A_{ij}\}(1+o(1))\\
\nonumber
&\leq 
\{2+|\log(\rho_n \|f\|_{\min})| +|\log(1-\rho_n \|f\|_0)|\}^2\frac{1}{n^4} 4 \|f\|_0\binom{n}{2}\\
&\leq 2\{2+|\log(\rho_n \|f\|_{\min})| +|\log(1-\rho_n \|f\|_0)|\}^2\frac{1}{n^2}  \|f\|_0.
\end{align}
As this ends up summing up over all edge variables, it is not a function of $h$.  The intuition behind our understanding of the decay in $n$ would instead be based on summing over the entire set of edge variables, as the Taylor expansion would be valid for any estimator such that it converges in probability to a constant, similar convergence rates could be derived for any consistent estimator of $\bar{f}_{ab}.$ This shows that $h$ does not impact the variance of $\widehat{H}^*(G)$.

\section*{B $\quad$ Proof of Theorem 4}\label{app-b}
\begin{proof}

The estimator for the entropy of the Chung--Lu Model is given as follows:
	\begin{align*}
	    \widehat{\mathrm{H}}(f^s)&=\frac{1}{{n \choose 2}}\sum_{i<j} -\{[{\rho}_n \hat{g}(\xi_i)\hat{g}(\xi_j)\log({\rho}_n \hat{g}(\xi_i)\hat{g}(\xi_j)]+[(1-{\rho}_n \hat{g}(\xi_i)\hat{g}(\xi_j))\log(1-{\rho}_n \hat{g}(\xi_i)\hat{g}(\xi_j))]\}.
	\end{align*}
\newline
$\widehat{\mathrm{H}}(f^s)$ has \textit{almost} the form of U-statistic given as:
\begin{equation*}
    U_n=\frac{1}{{n \choose 2}} \sum_{1\leq i < j\leq n} h(\xi_i,\xi_j),
\end{equation*}
with 
\begin{align*}
h(\xi_i,\xi_j)&= -\{[{\rho}_n \hat{g}(\xi_i)\hat{g}(\xi_j)\log({\rho}_n \hat{g}(\xi_i)\hat{g}(\xi_j))]+[(1-{\rho}_n \hat{g}(\xi_i)\hat{g}(\xi_j))\log(1-{\rho}_n \hat{g}(\xi_i)\hat{g}(\xi_j))]\}.
\end{align*}
We say \textit{almost} as $\hat{g}(\cdot)$ is a random function.
We then define:
\begin{align*}
    \widehat{\mathrm{H}}(f^s)&=\widehat{\mathrm{H}}(f^s)-\mathcal{E}(f^s)+\mathcal{E}(f^s),
\end{align*}
where 
\begin{align}\label{eq:confthm-e-def}
   {\cal  E}( f^s)&=\frac{1}{{n \choose 2}} \sum_{1\leq i < j\leq n} \widetilde{h}(\xi_i,\xi_j),
\end{align}
with
\begin{align*}
    \widetilde{h}(\xi_i,\xi_j)&=-\{\rho_n {g}(\xi_i){g}(\xi_j)\log({\rho}_n {g}(\xi_i){g}(\xi_j))]+[(1-{\rho}_n {g}(\xi_i){g}(\xi_j))\log(1-{\rho}_n {g}(\xi_i){g}(\xi_j))]\}.
\end{align*}
We now need to show that $\widehat{\mathrm{H}}(f^s)-\mathcal{E}(f^s)=o_P(1)$. We give two auxillary lemmas that will help us achieve our goal.
We also note that $\widehat{\mathrm{H}}(f^s)-\mathrm{H}(f^s)=\widehat{\mathrm{H}}(f^s)-\mathcal{E}(f^s)+\mathcal{E}(f^s)-\mathrm{H}(f^s)$. We will now determine the distribution of $\widehat{\mathrm{H}}(f^s)-\mathcal{E}(f^s)$ and show that $\mathcal{E}(f^s)-\mathrm{H}(f^s)=o(1)$, so that the distribution of $\widehat{\mathrm{H}}(f^s)-\mathrm{H}(f^s)$ will be determined by $\widehat{\mathrm{H}}(f^s)-\mathcal{E}(f^s)$.

\begin{lemma}
Let $\widehat{g}(\xi_i)=g(\xi_i)+Y_i$, where $Y_i=o_P(1).$
\end{lemma} 
\begin{proof}
\begin{align}
   \nonumber \widehat{g}(\xi_i)&=g(\xi_i)+Y_i =\frac{d_i}{\sqrt{\|d\|_1}} \\&=\frac{d_i}{\sqrt{\mathbb{E}[\|d\|_1]}}\underbrace{\Big\{\sqrt{\frac{\mathbb{E}[\|d\|_1]}{\|d\|_1}}\Big\}}_{{\overset{P}{\to 1}}} \label{lemmaeq}
\end{align}
We already know that $\mathbb{E}[Y_i]=0,$
so for determining the order of $Y_i$ we  calculate $\mathrm{Var}(Y_i)$ and then use Chebyshev's inequality. 
\begin{align*}
    \mathrm{Var}(Y_i)&=\mathrm{Var}(\widehat{g}(\xi_i)-g(\xi_i))\\
    &=\mathrm{Var}_{\xi}\mathbb{E}_{A|\xi}[\widehat{g}(\xi_i)-g(\xi_i)|\xi]+\mathbb{E}_{\xi}\mathrm{Var}_{A|\xi}\Big(\frac{d_i}{\sqrt{\|d\|_1}}|\xi\Big)\\
    &=\mathbb{E}_{\xi}\mathrm{Var}_{A|\xi}\Big(\frac{d_i}{\sqrt{\|d\|_1}}|\xi\Big)
\end{align*}

Rewriting \eqref{lemmaeq} we get:

\begin{align*}
    &\frac{d_i}{\sqrt{\mathbb{E}[\|d\|_1]}}\underset{\overset{P}{\to 1}}{\Big\{\sqrt{\frac{\mathbb{E}[\|d\|_1]}{\|d\|_1}}}\Big\} \\
    &=
    \frac{d_i}{\sqrt{\mathbb{E}[\|d\|_1|\xi]}\underbrace{\sqrt{1+\frac{\|d\|_1-\mathbb{E}[\|d\|_1|\xi]}{\mathbb{E}[\|d\|_1|\xi]}}}_{\sqrt{1+Z_n}}}
\end{align*}
To determine the order of $Z_n$ with $\mathbb{E}[Z_n]=0$, we again need to calculate $\mathrm{Var}(Z_n)$. We start by noting:
\begin{align*}
    &\mathrm{Var}(\|d\|_1-\mathbb{E}[\|d\|_1]|\xi)
    \\&=\mathrm{Var}(\|d\|_1|\xi)\\
    &=\mathrm{Var}\sum_{i}\sum_{j\neq i} A_{ij}|\xi \\
    &=\sum_{i}\sum_{j\neq i} \rho_n g(\xi_i)g(\xi_j)(1-\rho_n g(\xi_i)g(\xi_j)) \\
    & \leq g_{max}^2 \rho_n n(n-1),
\end{align*}
uniformly in $\xi$ where $g_{max}=\sup_{x} g(x)$.
For the denominator we have:
\begin{align*}
    \mathbb{E}[\|d\|_1|\xi]=\sum_{i=1}^n \sum_{j \neq i} \rho_n g(\xi_i) g(\xi_j) \geq \rho_n n(n-1) g_{min}^2
\end{align*}
Hence, for $\mathrm{Var}(Z_n)$ we obtain
\begin{align*}
    \mathrm{Var}(Z_n)& \leq \frac{g_{max}^2 \rho_n n(n-1)}{\rho_n^2 n^2 (n-1)^2 g_{min}^2}=\frac{g_{max}^2}{\rho_n n(n-1) g_{min}^2}.
\end{align*}
As $Z_n=o_P(1)$, by the Stochastic Taylor Expansion we get
\begin{align*}
    Y_i=\frac{d_i}{\sqrt{\mathbb{E}[\|d\|_1|\xi]}}\Big\{1+\frac{1}{2}Z_n+O_P(\sqrt{\mathrm{Var}(Z_n))}\Big\}.
\end{align*}
Assuming $Z_n$ is uniformly integrable,

\begin{align*}
    \mathrm{Var}(Y_i)&=\mathrm{Var}(Y_i|\xi_i)=\mathrm{Var}(\frac{d_i}{\sqrt{\mathbb{E}[\|d\|_1|\xi]}}|\xi).
\end{align*}
We already know the order of $\sqrt{\mathbb{E}[\|d\|_1|\xi]}$, so now we calculate $\mathrm{Var}(d_i|\xi)$ as follows
\begin{align*}
    \mathrm{Var}(d_i|\xi)&=\mathrm{Var}(\sum_{j\neq i} A_{ij}|\xi)=\sum_{j \neq i} \mathrm{Var}(A_{ij}|\xi)\\
    &=\sum_{j \neq i} \rho_n g(\xi_i)g(\xi_j)(1-\rho_n g(\xi_i)g(\xi_j))\\
    & \leq \rho_n (n-1)g_{max}^2.
\end{align*}
So, 
\begin{align*}
    \mathrm{Var}(Y_i)& \leq \frac{1}{\rho_n g_{min}^2 n(n-1)} \rho_n (n-1)g_{max}^2=\mathcal{O}\Big(\frac{1}{n}\Big),
\end{align*}
which in limit tends to zero, i.e.,
\begin{align*}
    \lim_{n \to \infty}\mathrm{Var}(Y_i) \to 0.
\end{align*}
We need $k=\mathcal{O}(\sqrt{n})$, so we choose $k$ to be $\sqrt{\mathbb{E}[d_i|\xi]}$. Then, by Chebyshev's inequality we have:
\begin{align*}
    \lim_{n \to \infty} P\Big\{|Y_i|\geq \mathcal{O}\Big(\frac{1}{\sqrt{n}}\Big)\Big\}\leq \frac{1}{\mathbb{E}[d_i|\xi]},
\end{align*}
which corresponds to 
\begin{align*}
    \lim_{n \to \infty} P(|Y_i|\geq 0) = 0.
\end{align*}
Therefore, $Y_i=o_P(1)$.
\end{proof}

\begin{lemma} Referring back to the statement of $\mathcal{E}(f)$ in~\eqref{eq:confthm-e-def}, we have the following for $\Delta\mathcal{E}(f)=o_P(1)$.
\begin{equation*}
    \widehat{\mathrm{H}}(f^s)=\mathcal{E}(f^s)+\Delta \mathcal{E}(f^s)
\end{equation*}
\end{lemma}

\begin{proof}
As stated before, we know that $\widehat{g}(\xi_i)=g(\xi_i)+Y_i$ with $Y_i=o_P(1)$. Rewriting $\widehat{\mathrm{H}}(f^s)$ in terms of $\widehat{g}(\xi_i)=g(\xi_i)+Y_i$, we get:
\begin{align*}
    \widehat{\mathrm{H}}(f^s) 
    &=-\frac{1}{{n \choose 2}}\Big\{\sum_{i,j}\underbrace{\rho_n (g(\xi_i)+Y_i)(g(\xi_j)+Y_j)\log((g(\xi_i)+Y_i)(g(\xi_j)+Y_j))}_{\text{term (1)}}\\
    \begin{split}
    &-\underbrace{[(1-\rho_n(g(\xi_i)+Y_i)(g(\xi_j)+Y_j))\log(1-\rho_n(g(\xi_i)+Y_i)(g(\xi_j)+Y_j))]}_{\text{term (2)}}\Big\}
    \end{split}
    .
    \end{align*}

After simple algebra and using identities for the logarithmic function;
we arrive to the conclusion that:
\begin{align*}
    \widehat{\mathrm{H}}(f^s)=\mathcal{E}(f^s)+\Delta \mathcal{E}(f^s),
\end{align*}
with 
\begin{align*}
    \Delta \mathcal{E}(f^s)= \mathcal{O}_P(Y_i^4).
\end{align*}
As $Y_i=o_P(1)$, and every term in $\Delta \mathcal{E}(f^s)$ has at least one factor of $Y_i$, then $\Delta \mathcal{E}(f^s)=o_P(1)$.
\end{proof}

So, the distribution of $\widehat{\mathrm{H}}(f^s)$ is determined only by $\mathcal{E}(f^s)$ and its variance is calculated via the classic theory of U-statistics as follows:
\begin{equation}\label{var_CM}
\sigma^2_{\widehat{\mathrm{H}}}=\frac{1}{{n\choose 2}} \Big\{ 2(n-2)\zeta_1+\zeta_2\Big\},
\end{equation}
with
$$\zeta_k=\mathrm{Var}(\Tilde{h}_k(\xi_i,\xi_j)),$$ where 
$$\Tilde{h}_k(x_1,\cdots,x_k)=\mathbb{E}[\Tilde{h}(X_1,\ldots,X_m)|X_1=x_1,\ldots,X_k=x_k].$$ 
\newline
Asymptotically, the leading contribution comes from $\zeta_1$. Thus, we let ${\sigma^*}^{2}_{\widehat{\mathrm{H}}}=\frac{1}{{n \choose 2}}2(n-2)\zeta_1$. For $\zeta_1$, we have that:
\begin{align*}
\zeta_1&=
(\rho_n\log \rho_n)^2\int_0^1 (1-g(x))^2 \mathrm{d}x \\
&+2\rho_n^2 \log \rho_n \int_0^1 (1-g(x))\Big[I_1(g(x)-1)+g(x)(\log g(x)-1)+1 \Big] \mathrm{d}x\\
&+\rho_n^2\{\|g\|_2+2\|g\|_2I_1+\int_{[0,1]}(g(y))^2\log g(y) \mathrm{d}y)(1+I_1)+\int_{[0,1]}(g(y)\log g(y))^2\mathrm{d}y\},
\end{align*}
where $I_1=\int_0^1 g(y)\log g(y) \mathrm{d}y$
and we define the function $p-$norm to be \[g_p=\|g\|_p=\int_0^1 g^p (x) \ \mathrm{d}x,\]
for $p=1,2,3,4.$

\end{proof}
\begin{lemma}$\mathcal{E}(f^s)-\mathrm{H}(f^s)=o_P(1)$.
\end{lemma}
\begin{proof}
    Recall that \[\mathcal{E}(f^s)=-\frac{1}{{n \choose 2}}\sum_{i=1}^n\sum_{j=1}^n \{\rho_n g(\xi_i)g(\xi_j)\log(\rho_n g(\xi_i)g(\xi_j))+(1-\rho_n g(\xi_i)g(\xi_j))\log(1-\rho_n g(\xi_i)g(\xi_j))\},\] and \[\mathrm{H}(f^s)=-\iint_{[0,1]^2}\rho_n g(x)g(y)\log(\rho_n g(x)g(y))+(1-\rho_n g(x)g(y))\log(1-\rho_n g(x)g(y))\} \ \mathrm{d}x \mathrm{d}y.\]
    We note that $\mathrm{H}(f^s)=\mathbb{E}_{\xi}[\mathcal{E}(f^s)].$ We note the form of the variance from Equation~\ref{var_CM} is of the order $O((\rho_n\log\rho_n)^2/n)$. By the Chebyshev's inequality, for any $\epsilon>0$ and $C>0$, we then obtain
    \begin{equation}
        \lim_{n\to \infty}\mathbb{P}(|\mathrm{H}(f^s)-\mathcal{E}(f^s)|>\epsilon)\leq \frac{C\rho_n}{n{\epsilon}^2}\to 0,
    \end{equation}
    which concludes the proof.
\end{proof}

\section*{C.1 $\quad$ Entropy of a Block--Constant Graphon}\label{app-c.1}
Given the block constant graphon function as defined in Definition 3
\[W^b(x,y)=\sum_{a,b}\theta_{ab} \mathds{1}(x \in I_a)\mathds{1}(y \in I_b),\] 
its entropy is given as
\[{\mathrm{H}}(W^b)=\iint_{[0,1]^2} h(W^b(x,y))\, \mathrm{d}x \mathrm{d}y,\]
where $h(\cdot)$ is the binary entropy function defined in (2.2).
Partitioning the unit interval into $k$ subintervals, we have 
$I_a=[u_{a-1},u_a)$ for $a=1,\ldots,k-1$ and $I_=[u_{k-1},1]$ with $0=u_0<u_1<\ldots<u_{k-1}<1$, 
then we get
\begin{align}\label{eq-sbm-entropy} {\mathrm{H}}(W^b)    
& \nonumber =-\iint_{[0,1]^2}\Bigg\{\sum_{a=1}^k\sum_{b=1}^k\theta_{ab} \mathds{1}(x \in I_a) \mathds{1}(y \in I_b)h(\theta_{ab}) \}
\,\mathrm{d}x\mathrm{d}y \Bigg\}
\\
& \nonumber =-\sum_{a=1}^k \sum_{b=1}^k h(\theta_{ab})\int_{u_{a-1}}^{u_a} \int_{u_{b-1}}^{u_b}
\,\mathrm{d}x \mathrm{d}y
\\
&=-\sum_{a=1}^k \sum_{b=1}^k h(\theta_{ab})(u_{a}-u_{a-1})(u_{b}-u_{b-1}).
\end{align}
If wlog we consider the $k$ subintervals to be equal, then it simplifies to:
\begin{equation*}
    H(W^b)=-\frac{1}{k^2} \sum_{a=1}^k \sum_{b=1}^k h(\theta_{ab}).
\end{equation*}

\section*{C.2 $\quad$ Proof of Corollary 5}\label{app-c.2}
We denote a single entry of the connection probability matrix by $\theta$ and the matrix by $\Theta$. Following~\cite{klopp2017oracle}, we denote $\widehat{\Theta}$ the least-squares estimator of $\Theta$ in the collection of all $k\times k$ block matrices with block size larger than some integer $n_0$. We denote the set of all mappings $z:[n]\to [k]$ with minimum block size $n_0$ by $\mathcal{Z}_{n,k,n_0}$. For any $z \in \mathcal{Z}_{n,k,n_0}$, $Q\in \mathbb{R}^{k \times k}_{\mathrm{sym}}$, the residual sum of squares is defined as
\begin{equation*}
    L(Q,z)=\sum_{(a,b)\in[k]\times [k]}\sum_{(i,j)\in z^{-1}(a) \times z^{-1}(b),j<i}(A_{ij}-Q_{ij})^2,
\end{equation*}
and the least squares estimator of $(Q,z)$ as
\begin{equation*}
    (\widehat{Q},\hat{z}) \in \mathrm{arg}\min_{\mathbf{Q}\in \mathbb{R}^{k \times k}_{\mathrm{sym}},z\in \mathcal{Z}_{n,k,n_0}} L(\mathbf{Q},z).
\end{equation*}
The block constant least squares estimator of $(\Theta)_{ab}$ is defined as $\widehat{\Theta}_{ab}=\widehat{Q}_{\hat{z}(i) \hat{z}(j)}$ for all $1\leq a,b \leq k$, where $\hat{z}(i)=a$ and $\hat{z}(j)=b$, respectively. Note that $\widehat{\Theta}_{ab}\in[0,1]$.

We have seen the form that the entropy of the block-constant graphon takes after further simplifications in Appendix~\ref{app-c.1}. 
Having defined the estimator of $\Theta$, we denote the estimated entropy by \begin{equation*}
   \widehat{\mathrm{H}}(\theta)=-\frac{1}{k^2}\sum_{a<b}^k[\hat{\theta}_{ab}\log(\hat{\theta}_{ab})+(1-\hat{\theta}_{ab})\log(1-\hat{\theta}_{ab})]
\end{equation*}
Following the proof sketch of Theorem 2 and~\cite{klopp2017oracle}, for  we obtain 
\begin{align*}
    &|\widehat{\mathrm{H}}(f^b)-\mathrm{H}(f^b)|\\
    &\leq C\left(\rho_n\frac{k^2}{n^2}+\rho_n\frac{\log k}{n}+\rho_n^2\sqrt{\frac{k}{n}}\right)^{1/2}+C\left(\rho_n\frac{k^2}{n^2}+\rho_n\frac{\log k}{n}+\rho_n^2\sqrt{\frac{k}{n}}\right),
\end{align*}
with probability at least $1-\exp(-C'n)$, where $C,C'>0$.

\section*{D.0 $\quad$ Preliminaries}
\subsection*{Exchangeable graphs and graphons}
 Exchangeability is a form of probabilistic invariance indicating that the distribution of a graph's adjacency matrix remains unchanged under permutations of its nodes. We consider a graph $G$ to be exchangeable if its adjacency matrix is a jointly exchangeable array. This class of graphs is well-studied, see for example~\cite{lauritzen2008exchangeable,orbanz2014bayesian,austin2015exchangeable,janson2013graphons}, and is particularly important, as it can be characterized by a representation theorem called the Aldous-Hoover theorem~\cite{aldous1981representations,hoover1979relations}. 
\begin{definition}[Joint Exchangeability]
A random array $(A_{ij})_{i,j \in \mathbb{N}^*}$ is jointly exchangeable if for any permutation $\pi$ of $\mathbb{N}^*$ 
\begin{equation}
(A_{ij})\overset{d}{=}(A_{\pi(i)\pi(j)}).
\end{equation}
\end{definition}

\begin{theorem}[Aldous--Hoover~\cite{aldous1981representations,hoover1979relations}]\label{thm-aldous}
Let $A$ be a jointly exchangeable random array. Then there exists an i.i.d. sequence $\xi=(\xi_1,...,\xi_n)$ following $U(0,1)$, a random variable $\gamma \sim U(0,1)$ independent of $\xi$, and a function $W:[0,1]^3 \to [0,1]$ such that \begin{equation}\label{graphon}
    Pr(A_{ij}=1 | \xi, \gamma) = W(\xi_i,\xi_j,\gamma),
\end{equation} and $A_{ij}$ are conditionally independent across $i,j$ given $\xi$ and $\gamma$.
\end{theorem}\label{thm-aldous-hoover}
To depict any infinite array, an extra random variable $\gamma$ is necessary. However, with only a single realization of $A$, measuring the dependency on $\gamma$ is impossible, so this dependency is suppressed under the assumption of dissociation. We use the notation $W(\xi_i, \xi_j)$ for $1 \leq i,j \leq n = |V|$ to denote the unscaled graphon.

Graph limit theory was developed for dense graphs, but empirical studies show that most real-world graphs are sparse~\cite{newman2018networks}. Following~\cite{orbanz2014bayesian}, we define dense graphs as those with $\Theta(n^2)$ edges, and sparse graphs otherwise. To account for sparsity, a common approach~\cite{bollobas2009metrics} is to introduce $\rho_n > 0$, the scale parameter representing the expected proportion of present edges. An alternative sparse representation can be derived by modifying Equation~\eqref{graphon} to yield a random graph with $O(\rho_n n^2)$ edges as follows:

\begin{equation}\label{sparse-graphon}
    Pr(A_{ij}=1|\xi)=\rho_n W(\xi_i,\xi_j)=f(\xi_i,\xi_j),
\end{equation}
where the usual assumption is that $\rho_n \to 0$ as $n \to \infty$. 
This model is thought of as a graphon model that has been sparsified with probability $1-\rho_n$ and has been studied in~\cite{bickel2009nonparametric,bickel2011method,wolfe2013nonparametric,xu2014edge,klopp2017oracle}.
We employ~\eqref{sparse-graphon} as a method to study estimation for graphs with a large but finite number of nodes because it does not make sense in the limit, and we will denote the \emph{graphon} by $f(\cdot,\cdot)=\rho_n W(\cdot,\cdot)$ to account for all sparsity regimes. The dense case corresponds to $\rho_n=1$.

\subsection{Nonparametric graphon estimation}
\label{sec-graphon-est}
Estimating the graphon has remained a problem of intense scrutiny~\cite{Olhede_2014,bickel2009nonparametric,chatterjee2015matrix,gao2015rate}. A commonly adopted method is approximating a graphon by a stochastic block model~\cite{Olhede_2014,wolfe2013nonparametric,gao2015rate,klopp2017oracle}. Authors of~\cite{wolfe2013nonparametric,gao2015rate} show that whenever $f$ is a  H\"older continuous function, it can be consistently estimated from $A$, and authors of~\cite{gao2015rate} provide an optimal rate of convergence for the dense regime, while authors of~\cite{klopp2017oracle} complement the optimal rates for the sparse regime as well.

 This assumption is now the standard in nonparametrics literature which is necessary for making estimation tractable. We refer the reader to~\cite{Olhede_2014,gao2015rate,klopp2017oracle,wolfe2013nonparametric} for details and convergence rates of nonparametric graphon estimation, both in the dense and the sparse regime. We now detail graphon estimation via Stochastic Block Model approximation as also treated by~\cite{Olhede_2014}. 

 Given a single adjacency matrix $A\in \{0,1\}^{n \times n}$, we estimate $f(x,y)$ up to the rearrangement of the axes, using a stochastic block model approximation with a single community size $h$, which can be either determined by the user or automatically determined from the adjacency matrix. Given this notation, then we have $n=hk+r$, for integers $h,k,r$, where $k=\lfloor n/h \rfloor$, is the number of blocks, $h\in\{2,\ldots,n\}$ is the bandwidth and $r=n \ \mathrm{mod} \ h$ is a reminder term between $0$ and $h-1$. We introduce a community membership vector $z$ which groups together nodes that should lie in the same groups. Its components will be in $\{1,\ldots,k\}$. 
 Let $\mathcal{Z}_k \subseteq \{1,\ldots,k\}^n$ represent the set including all community assignments $z$ respecting the form of $n=hk+r$. The main difficulty lies in estimating $z$, and several methods are available for estimating $z$. Common methods include spectral clustering~\cite{von2007tutorial}, likelihood methods~\cite{celisse2012consistency, bickel2009nonparametric, amini2013pseudo,choi2012stochastic, zhao2012consistency}, and modularity maximization~\cite{newman2004finding}.
 We estimate $z$ by the method of maximum likelihood as follows,
 \begin{equation}\label{z}
     \hat{z}=\mathrm{argmax}_{z\in \mathcal{Z}_k}\sum_{i<j}A_{ij}\log \Bar{A}_{z_i z_j}+(1-A_{ij})\log(1-\Bar{A}_{z_i z_j})\},
 \end{equation}
 where for all $1\leq a,b\leq k$ the histogram bins are defined as follows
 \begin{equation}\label{A_bar}
\Bar{A}_{ab}=\frac{\sum_{i<j}A_{ij}\mathds{1}(z_i=a)\mathds{1}(z_j=b)}{\sum_{i<j}\mathds{1}(z_i=a)\mathds{1}(z_j=b)}.
 \end{equation}
 Each bin height resembles the probability of edges present in the histogram bin which corresponds to a block of Bernoulli trials given by the objective function in~\ref{z}.
 It is also evident that $\Bar{A}_{ab}=\Bar{A}_{ba}$.
 The network histogram~\cite{Olhede_2014} is given as
 \begin{equation}\label{est-graphon}
     \hat{f}(x,y;h)=\hat{\rho}_n^+ \Bar{A}_{\min\{\lceil nx/h\rceil,k\}\min\{\lceil ny/h\rceil,k\}}, \quad 0<x,y<1,
 \end{equation}
 where $\hat{\rho}_n^+$ denotes the generalized inverse of $\hat{\rho}_n$, and $\hat{\rho}_n=\frac{\sum_{i<j}A_{ij}}{{n \choose 2}}$.
\subsection{Entropy of an Exchangeable Random Graph} 
Here, we briefly revisit the theoretical relationship between graphon entropy and the entropy of an exchangeable random graph sampled from a graphon, as discussed in~\cite{janson2013graphons}, without addressing estimation.

Let $G(n, W)$ represent a random graph generated by the graphon $W$, with adjacency matrix $A \in \{0,1\}^{n \times n}$ encoding the ${n \choose 2}$ possible edges. For finite $n$, $G(n, W)$ is a discrete random variable with $2^{{n \choose 2}}$ possible outcomes. While $\mathrm{H}(G(n))$ can be directly computed for independent edges $A_{ij} \overset{iid}{\sim} \mathrm{Bernoulli}(p_{ij})$, dependencies in $G(n, W)$ complicate this computation. In $G(n, W)$, edges $A_{ij}$ become conditionally independent given the latent vector $\xi$, as outlined later in Theorem~\ref{thm-aldous}. The conditional entropy of $G(n, W)$ given $\xi$, where $\xi_i \sim U(0,1)$, is expressed as~\cite{janson2013graphons}:
\begin{align}\label{graph-entropy}
\nonumber \mathrm{H}(G(n,W)|\xi_1,\ldots,\xi_n)&=\sum_{i<j} \mathrm{H}(A_{ij}|\xi_i,\ldots,\xi_n)=\sum_{i<j}\mathrm{H}(\mathrm{Bernoulli}(p_{ij}))\
&=\sum_{i<j}h(p_{ij})=\sum_{i<j}h(W(\xi_i,\xi_j)).
\end{align}
We restate the asymptotic properties of graph entropy as presented in Theorem D.5 of~\cite{janson2013graphons}.
\begin{theorem}[\cite{janson2013graphons}]\label{entropy-asymp}
    Let $G(n,W)$ be an exchangeable random graph on $n$ nodes generated by the graphon $W$ from a probability space $(\Omega,\mu)$. Then as $n\to \infty$,
    \begin{equation}
        \frac{\mathrm{H}(G(n,W))}{{n \choose 2}} \to \iint_{\Omega}h(W(x,y))\mathrm{d}\mu(x)\mathrm{d}\mu(y).
    \end{equation}
\end{theorem}

 In our setting, we set $\Omega=[0,1]^2$ and $\mu$ to be $U[0,1]$ wlog. Theorem~\ref{entropy-asymp} shows that the normalized entropy of the exchangeable graph converges to the entropy of the graphon that generated it as the number of nodes increases, therefore also relating the graphon entropy to graph entropy. 
\end{document}